\documentclass[draftcls,12pt,onecolumn]{IEEEtran}
\usepackage{cite}% *** CITATION PACKAGES ***
\usepackage[cmex10]{amsmath}% *** MATH PACKAGES ***
\usepackage{amssymb}
\usepackage{amsthm}
\usepackage{mathrsfs}%SF
\usepackage{bm}
\usepackage[mathscr]{eucal}
\usepackage{amssymb,amsmath,amsthm,amsfonts,latexsym}
\usepackage{amsmath,graphicx,bm,xcolor,url}
\usepackage{graphicx}
\usepackage{latexsym}%SF
\usepackage{CJK}%SF
\usepackage{indentfirst}%SF
\usepackage{geometry}%SF
\usepackage{psfrag}%SF
\usepackage{setspace}%SF
\usepackage{algorithmic}% *** SPECIALIZED LIST PACKAGES ***
\usepackage{algorithmic, cite}
\usepackage{algorithm}
\usepackage{array}% *** ALIGNMENT PACKAGES ***
\usepackage{mdwmath}
\usepackage{mdwtab}
\usepackage{eqparbox}
\usepackage{url}% *** PDF, URL AND HYPERLINK PACKAGES ***
\usepackage{epstopdf}
\usepackage{epsfig,epsf,psfrag}
\usepackage{fixltx2e}%ordering of single and double column floats
\usepackage{verbatim}
\usepackage{textcomp}
\usepackage{psfrag} %%25 June

\usepackage{caption}
\usepackage{subcaption}%2013-03-16

\theoremstyle{plain}
\newtheorem{mythe}{Theorem}
\theoremstyle{remark}
\newtheorem{mylem}{Lemma}
\theoremstyle{plain}

\theoremstyle{remark}
\newtheorem{mypro}{Proposition}
\theoremstyle{plain}

\theoremstyle{remark}
\newtheorem{myrem}{Remark}

\theoremstyle{remark}

\theoremstyle{remark}

\theoremstyle{remark}

\theoremstyle{remark}

\catcode`~=11 \def\UrlSpecials{\do\~{\kern -.15em\lower .7ex\hbox{~}\kern .04em}} \catcode`~=13

\allowdisplaybreaks[4]

\newcommand{\calC}{\mathcal{C}}

\newcommand{\calN}{\mathcal{N}}

\newcommand{\calQ}{\mathcal{Q}}

\newcommand{\bbE}{\mathbb{E}}

\DeclareMathAlphabet{\mathbsf}{OT1}{cmss}{bx}{n}
\DeclareMathAlphabet{\mathssf}{OT1}{cmss}{m}{sl}% slanted sans serif

\DeclareSymbolFont{bsfletters}{OT1}{cmss}{bx}{n}
\DeclareSymbolFont{ssfletters}{OT1}{cmss}{m}{n}
\DeclareMathSymbol{\bsfGamma}{0}{bsfletters}{'000}
\DeclareMathSymbol{\ssfGamma}{0}{ssfletters}{'000}
\DeclareMathSymbol{\bsfDelta}{0}{bsfletters}{'001}
\DeclareMathSymbol{\ssfDelta}{0}{ssfletters}{'001}
\DeclareMathSymbol{\bsfTheta}{0}{bsfletters}{'002}
\DeclareMathSymbol{\ssfTheta}{0}{ssfletters}{'002}
\DeclareMathSymbol{\bsfLambda}{0}{bsfletters}{'003}
\DeclareMathSymbol{\ssfLambda}{0}{ssfletters}{'003}
\DeclareMathSymbol{\bsfXi}{0}{bsfletters}{'004}
\DeclareMathSymbol{\ssfXi}{0}{ssfletters}{'004}
\DeclareMathSymbol{\bsfPi}{0}{bsfletters}{'005}
\DeclareMathSymbol{\ssfPi}{0}{ssfletters}{'005}
\DeclareMathSymbol{\bsfSigma}{0}{bsfletters}{'006}
\DeclareMathSymbol{\ssfSigma}{0}{ssfletters}{'006}
\DeclareMathSymbol{\bsfUpsilon}{0}{bsfletters}{'007}
\DeclareMathSymbol{\ssfUpsilon}{0}{ssfletters}{'007}
\DeclareMathSymbol{\bsfPhi}{0}{bsfletters}{'010}
\DeclareMathSymbol{\ssfPhi}{0}{ssfletters}{'010}
\DeclareMathSymbol{\bsfPsi}{0}{bsfletters}{'011}
\DeclareMathSymbol{\ssfPsi}{0}{ssfletters}{'011}
\DeclareMathSymbol{\bsfOmega}{0}{bsfletters}{'012}
\DeclareMathSymbol{\ssfOmega}{0}{ssfletters}{'012}

\newcommand{\hatB}{\widehat{B}}

\newcommand{\tilc}{\widetilde{c}}

\newcommand{\tilf}{\widetilde{f}}

\newcommand{\tilg}{\widetilde{g}}

\newcommand{\tilh}{\widetilde{h}}

\newcommand{\tilP}{\widetilde{P}}

\newcommand{\hatR}{\widehat{R}}

\newcommand{\tilR}{\widetilde{R}}

\newcommand{\tils}{\widetilde{s}}

\newcommand{\tilw}{\widetilde{w}}

\newcommand{\tily}{\widetilde{y}}

\def\norm#1{\left\| #1 \right\|}
\def\norm2#1{\left\| #1 \right\|_2}
\def\norm22#1{\left\| #1 \right\|_2^2}

\newcommand{\eqa}{\stackrel{(a)}{=}}
\newcommand{\eqb}{\stackrel{(b)}{=}}
\newcommand{\eqc}{\stackrel{(c)}{=}}
\newcommand{\eqd}{\stackrel{(d)}{=}}

\newcommand{\qednew}{\nobreak \ifvmode \relax \else
      \ifdim\lastskip<1.5em \hskip-\lastskip
      \hskip1.5em plus0em minus0.5em \fi \nobreak
      \vrule height0.75em width0.5em depth0.25em\fi}

\renewcommand{\baselinestretch}{1.33}
\newcommand{\tilmu}{\tilde{\mu}}
\newcommand{\tilsigma}{\tilde{\sigma}}
\newcommand{\tilepsilon}{\tilde{\epsilon}}

\title{Modulation in the Air: Backscatter Communication over Ambient OFDM Carrier}
\author{Gang~Yang, \emph{Member, IEEE}, Ying-Chang~Liang, \emph{Fellow, IEEE}, \\ 
Rui Zhang, \emph{Fellow, IEEE}, and Yiyang~Pei, \emph{Member, IEEE}
\thanks{G.~Yang is with the National Key Laboratory of Science and Technology on Communications, University of Electronic Science and Technology of China (UESTC), Chengdu, P. R. China (e-mail: yanggang@uestc.edu.cn).}
\thanks{Y.-C.~Liang is with University of Electronic Science and Technology of China (UESTC), Chengdu, P. R. China (e-mail: liangyc@ieee.org). He is also with the Institute for Infocomm Research (I$^2$R), A*STAR, Singapore.}
\thanks{R. Zhang is with the ECE Department, National University of Singapore (email: elezhang@nus.edu.sg).}
\thanks{Y. Pei is with Singapore Institute of Technology (email: yiyang.pei@singaporetech.edu.sg).}
\thanks{A preliminary version of this paper has been presented in \cite{YangLiangGC16}.}}

\begin{document}
\maketitle %Automatic title!
\vspace{-1.4cm}
\begin{abstract}
Ambient backscatter communication (AmBC) enables radio-frequency (RF) powered backscatter devices (BDs) (e.g., sensors, tags) to modulate their information bits over ambient RF carriers in an over-the-air manner. This technology also called ``modulation in the air'', thus has emerged as a promising solution to achieve green communications for future Internet-of-Things. This paper studies an AmBC system by leveraging the ambient orthogonal frequency division multiplexing (OFDM) modulated signals in the air. We first model such AmBC system from a spread-spectrum communication perspective, upon which a novel joint design for BD waveform and receiver detector is proposed. The BD symbol period is designed to be in general an integer multiplication of the OFDM symbol period, and the waveform for BD bit `0' maintains the same state within a BD symbol period, while the waveform for BD bit `1' has a state transition in the middle of each OFDM symbol period within a BD symbol period. In the receiver detector design, we construct the test statistic that cancels out the direct-link interference by exploiting the repeating structure of the ambient OFDM signals due to the use of cyclic prefix. For the system with a single-antenna receiver, the maximum-likelihood detector is proposed to recover the BD bits, for which the optimal threshold is obtained in closed-form expression. For the system with a multi-antenna receiver, we propose a new test statistic which is a linear combination of the per-antenna test statistics, and derive the corresponding optimal detector. The proposed optimal detectors require only knowing the strength of the backscatter channel, thus simplifying their implementation. Moreover, practical timing synchronization algorithms are proposed for the proposed AmBC system, and we also analyze the effect of various system parameters on the transmission rate and detection performance. Finally, extensive numerical results are provided to verify that the proposed transceiver design can improve the system bit-error-rate (BER) performance and the operating range significantly, and achieve much higher data rate, as compared to the conventional design.
\end{abstract}

\vspace{-0.5cm}
%\newpage

\section{Introduction}\label{introduction}
Wireless communication powered by ambient or dedicated radio-frequency (RF) source has drawn significant attention recently~\cite{BiZhangComMag15, BiZhangWirelessCom16}. In particular, ambient backscatter communication (AmBC) enables RF-powered backscatter devices (BDs) to harvest power from ambient RF signals (e.g., TV signal and WiFi signal), and to transmit information to nearby receivers (e.g., reader and smartphone) over the ambient RF carriers~\cite{ABCSigcom13}. Thus, it has drawn significant attention from both academia and industry recently. Since AmBC is carried out at the same frequency band as the ambient wireless system, they can be viewed as a spectrum sharing system \cite{YCLiangTVT15}. Different from traditional backscatter communication such as radio-frequency-identification (RFID) systems~\cite{Dobkinbook2007}, AmBC can exempt the reader from generating RF sinusoidal carriers, thus enables low-cost and energy-efficient ubiquitous communications. It has also been verified that the harvested power from ambient RF signals could be sufficient to power a high-throughput battery-less sensor~\cite{SampleSmith07, ParksSmith14}. Thus, AmBC is a promising technology for green communications with great potential for applications in next-generation Internet of Things (IoT)~\cite{Ishizaki11}.
%NikitinRao06,

Due to the spectrum sharing nature, an inherent characteristic of the AmBC system is that the receiver suffers from strong direct-link interference out of the remote RF source. There are existing literature on receiver design for AmBC which treat the direct-link interference as part of the background noise~\cite{ABCSigcom13, QianGaoAmBCTWC16, WangGaoAmBCTCOM16, WiFiBackscatter14}. In \cite{ABCSigcom13}, an averaging detector is proposed to decode the BD bits by treating the strong direct-link interference as noise, which results in very low decoding signal-to-noise-ratio (SNR) and thus low data rate. \textcolor{black}{In~\cite{QianGaoAmBCTWC16} and~\cite{WangGaoAmBCTCOM16}, maximum-likelihood (ML) detection is studied for an AmBC system in which the BD adopts differential modulation. However, for the scenario where the backscatter channel is much weaker than the direct-link channel, which is typical in practice due to BD's small reflection coefficient and double-attenuation in the backscatter link, the proposed detection schemes suffer from severe performance degradation.} An ambient WiFi backscatter system is proposed in \cite{WiFiBackscatter14}, in which a WiFi helper (e.g., smartphone) decodes the BD bits by detecting the changes in received signal strength indication (RSSI). However, this system has very low data rate and very limited communication range (less than one meter), since the detection of the RSSI changes suffers from the strong direct-link interference from the WiFi access point (AP).
%GaoWangICASSP16 the relative difference of

Moreover, there are other literature on receiver design for AmBC which cancel the direct-link interference by using signal processing methods~\cite{TurbochargingABCSigcom14, BackFiSigcom15,HitchHikeKattiSenSys16}. In~\cite{TurbochargingABCSigcom14}, the authors propose to use two receive antennas to cancel out the direct-link interference and thus increase the data rate. However, this interference cancellation scheme increases the complexity and cost of the receiver, and cannot be readily applied to a receiver with single or more than two antennas. A new WiFi backscatter system is proposed in \cite{BackFiSigcom15}, in which the WiFi AP decodes the received backscattered signal while simultaneously transmitting WiFi packages to a standard WiFi client. This design relies on the self-interference-cancellation for full-duplex radios, resulting in high complexity and cost, thus cannot be applied to commercial WiFi APs with complexity and cost constraints. \textcolor{black}{In~\cite{HitchHikeKattiSenSys16}, a backscatter device embeds its information on standard 802.11b packets by translating the originally transmitted 802.11b codeword to another valid 802.11b codeword. The detection WiFi AP first receives the original 802.11b packet decoded and sent by another WiFi AP, and then uses an XOR decoder to detect the information bit of the backscatter device. However, such detection relies on the collaboration of two WiFi APs, and this system is only applicable for the scenario of ambient 802.11b signal.}

\textcolor{black}{On the other hand, some recent work on AmBC address the problem of direct-link interference from the perspective of system design~\cite{PassiveWiFiNSDI16,InterscatterSigcom16, FSBackscatterSigcomm16}. A passive WiFi system is proposed in~\cite{PassiveWiFiNSDI16}, which requires a dedicated device to transmit RF sinusoidal carrier at a frequency that lies outside the desired WiFi channel, such that the WiFi receiver can suppress the resulting out-of-band (direct-link) carrier interference. However, the passive WiFi cannot be fully plugged-in and work with commodity devices, since it needs a dedicated device with specialized hardware to generate RF signal and perform carrier sensing. An inter-technology backscatter (interscatter) system is proposed in~\cite{InterscatterSigcom16}, which transforms wireless transmissions from one technology (e.g., Bluetooth) to another (e.g., WiFi) in the air. The BD creates frequency shifts on a single side of the carrier by using complex impedance of its backscatter circuit, so as to suppress the direct-link (carrier) interference and avoid a redundant copy on the opposite side of the carrier which exists in sideband modulation~\cite{PassiveWiFiNSDI16}. Similarly, a frequency-shifted backscatter (FS-Backscatter) system is proposed in~\cite{FSBackscatterSigcomm16} for on-body sensor applications, which reduces carrier interference by shifting the backscattered signal to a clean band that does not overlap with the carrier. However, the FS-Backscatter system cannot be applied to application scenarios where there is obvious temperature variability in the environment.}

Also, there are some recent work on traditional backscatter communication systems~\cite{HeWangCL11, BoyerSumit14, Carvalho14, BoaventuraCarvalho13, YangBackscatter15, KimionisSahalosTCOM14, HuangWPBCN16}, which use dedicated reader infrastructure to transmit RF sinusoidal carriers to power the passive tag and retrieve the tag information. The rate and reliability performance can be enhanced by using multi-antenna techniques. For instance, the bit-error-rate (BER) performance is analyzed for non-coherent frequency-shift-keying modulation in multi-input-single-output (MISO) fading channels~\cite{HeWangCL11}, and the diversity-multiplexing trade-off is investigated in~\cite{BoyerSumit14}. The reading range of RFID tags is improved by rectifier circuit design (see~\cite{Carvalho14} and references therein), special waveform design~\cite{BoaventuraCarvalho13} and energy beamforming~\cite{YangBackscatter15}. \textcolor{black}{In~\cite{HuangWPBCN16}, the authors propose a network architecture that integrates wireless power transfer and backscatter communications, called wirelessly powered backscatter communication networks, and optimize the network coverage probability and transmission capacity by applying stochastic geometry.} Recently, a bistatic architecture for backscatter communication is proposed in~\cite{KimionisSahalosTCOM14}, i.e., the low-cost carrier emitter is detached from the reader, to increase the reading range. While this bistatic backscatter communication system is similar to an AmBC system, this system requires to deploy and maintain a separate carrier emitter, which results in additional cost.

In this paper, we consider a new  AmBC system over ambient orthogonal frequency division multiplexing (OFDM) modulated carriers in the air. OFDM is a widely used modulation scheme in current wireless systems such as WiFi and DVB~\cite{GoldsmithWC2005}, thus is a readily available ambient RF source. We aim to investigate system modelling, transceiver design and performance analysis for such AmBC system. The main contributions of this paper are summarized as follows:
\begin{itemize}
  \item We first establish the system model for AmBC from a spread-spectrum (SS) communication perspective. We view the backscatter operation at the BD as an SS modulation technique, called ``modulation in the air'', i.e., the backscattered signal is viewed as the multiplication of a low-rate BD data signal and a high-rate spreading signal (i.e., the received ambient signal at the BD from the RF source) in an over-the-air manner. However, the ambient backscatter system is different from traditional SS systems in the sense that the spreading codes are unknown and time varying, and the receiver suffers from strong direct-link interference from the RF source.
      %(scaled by a constant reflection coefficient)
  \item Then, based on the established system model, we propose a novel joint design for BD waveform and receiver detector, which cancels out the direct-link interference. The BD symbol period is designed to be in general an integer multiplication of the OFDM symbol period, and the waveform for BD bit `0' maintains the same state within a BD symbol period, while the waveform for BD bit `1' has a state transition in the middle of each OFDM symbol period within a BD symbol period. With this new BD waveform design, we construct the test statistic for BD signal detection that cancels out the direct-link interference by exploiting the repeating structure of the ambient OFDM signals due to the use of cyclic prefix (CP). Our joint transceiver design cancels out direct-link interference without increasing the hardware complexity. To the best of our knowledge, this interference-cancellation scheme is proposed in this paper for the first time. %whole
  \item Furthermore, we investigate the optimal detector design for the receiver in the considered AmBC system. For the case of single-antenna receiver, based on the constructed test statistic, the ML detector is derived, for which the optimal detection threshold is obtained in closed-form expression. For the case of multi-antenna receiver, we propose a new test statistic for BD signal detection, which is a linear combination of the per-antenna test statistics, each constructed from the received signal at one receive antenna. The corresponding optimal detector is then derived. To perform optimal detection, the receiver requires to estimate only the strength of the backscatter channel. This exempts the receiver from estimating the complete information of the direct-link and backscatter channels, which is challenging for the AmBC system, due to unknown ambient signal and strong direct-link interference.
  \item To implement the designed AmBC system, we propose efficient methods for estimating the essential parameters required. Specifically, the BD uses an autocorrelation-based method to estimate the propagation delay for the source-to-BD channel in a blind manner, and we propose a new algorithm to perform timing synchronization at the receiver. The proposed algorithm exempts the receiver from knowing the synchronization preambles in the ambient OFDM signals.
  \item Moreover, we analyze the effect of various system parameters including the CP length, number of subcarriers, detection SNR and maximum channel spread, on the transmission rate and detection performance of the AmBC system.
  \item Finally, extensive numerical results are provided to show that the proposed transceiver design can achieve much lower BER and higher data rate than the conventional design~\cite{ABCSigcom13}. Also, numerical results show that the proposed timing synchronization methods are practically valid and efficient, and the deployment of multiple antennas at the receiver can improve the BER performance as well as the operating range significantly.
\end{itemize}

{\emph{Organizations}}: The rest of this paper is organized as follows. Section~\ref{SystemModel} presents the system model for AmBC over ambient OFDM carriers in the air. Section~\ref{sec:OptDetection} first presents the spread-spectrum based signal model, and then proposes the optimal joint transceiver design for the BD waveform and the receiver detector in a single-antenna system. Section~\ref{sec:parameter_est} presents practical methods for estimating essential parameters for implementing AmBC systems. Section~\ref{sec:MA_extension} further investigates the optimal receiver design for the case of multi-antenna receiver. Section~\ref{sec:detection_analysis} analyzes the effect of various system parameters on the transmission rate and detection performance. Section~\ref{sec:simulation} presents the numerical results. Finally, Section~\ref{conslusion} concludes the paper.
%With the designed transmitter,

{\emph{Notations}}: The main notations in this paper are listed as follows. $|\cdot|$ means the operation of taking the absolute value. $\otimes$ stands for the convolution operation of two signals. $\calN(\mu, \sigma^2)$ denotes the real Gaussian distribution with mean $\mu$ and variance $\sigma^2$. $\calC \calN(\mu, \sigma^2)$ denotes the circularly symmetric complex Gaussian (CSCG) distribution with mean $\mu$ and variance $\sigma^2$. $\calQ(x)$ denotes the $\calQ$-function, i.e., $\calQ(x)=\frac{1}{\sqrt{2 \pi}} \int_x^{\infty} e^{-r^2/2} dr$. $\bbE[\cdot]$ denotes the statistical expectation. $\lfloor \cdot \rfloor$ denotes the floor operation.

\section{System Model And Protocol Description} \label{SystemModel}
In this section, we present the system model and describe the link-layer protocol for the AmBC system over ambient OFDM carriers in the air. %, as well as the system protocol.
\subsection{System Model}
\begin{figure}[t!]
\vspace{-0.1cm}
\centering
\includegraphics[width=.65\columnwidth] {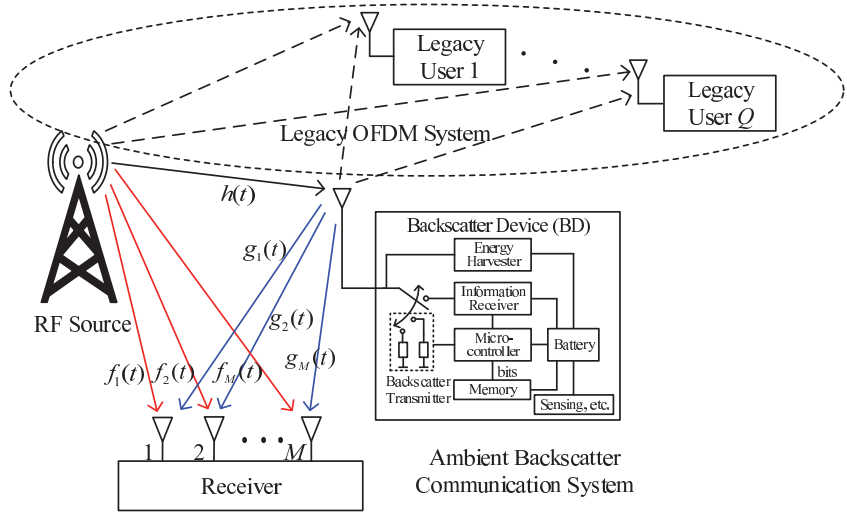}
\caption{System model for ambient backscatter communications over OFDM carriers in the air.}
% using received signals from a remote RF source system_structure_160802
\label{fig:Fig1}
\vspace{-0.8cm}
\end{figure}
As illustrated in Fig.~\ref{fig:Fig1}, we consider two co-existing communication systems: the legacy\footnote{Hereinafter, the term ``legacy'' refers to existing wireless communication systems, such as DVB, cellular and WiFi systems.} OFDM system which consists of an RF source (e.g., TV tower, WiFi AP) and its dedicated (legacy) users (e.g., TV receiver, WiFi client), and the AmBC system which consists of a BD with RF-power harvesting module and a receiver (e.g., reader in RFID systems) equipped with in general $M \ (M \geq 1)$ antennas. The RF source transmits OFDM signals to the legacy users. We are interested in the AmBC system in which the BD transmits its modulated signals to the receiver over the ambient OFDM carrier from the RF source. \textcolor{black}{The BD contains a single backscatter antenna, a backscatter transmitter (i.e., a switched load impedance), an energy harvester, an information receiver\footnote{\textcolor{black}{In practice, the simple BD can adopt the direct-conversion structure for information receiver, which is of low hardware-complexity, small size and low power~\cite{BryantRFIC2012}. It can further use ultra low-power analog-to-digital converter to reduce power consumption.}}, a micro-controller, a memory, a rechargeable battery replenished by the energy harvester, and other modules (e.g., sensing). The energy harvester collects energy from ambient OFDM signals and uses it to replenish the battery which provides power for all modules of the BD. To transmit information bits stored in the memory to the receiver, the BD modulates its received ambient OFDM carrier by intentionally switching the load impedance to change the amplitude and/or phase of its backscattered signal, and the backscattered signal is received and finally decoded by the receiver. Also, the BD antenna can be switched to the information receiver which is able to perform information decoding and other simple signal processing operations.}
%\footnote{The BD is typically equipped with a single antenna, due to the size and cost constraint in practice.}  via backscatter modulation

As shown in Fig.~\ref{fig:Fig1}, for the channel between the RF source and the BD, we denote $h(t)$ as the (baseband) channel impulse response (CIR) with multi-path spread $\tau_{\sf h}$, and $d_{\sf h}$ as the channel propagation delay which is the arrival time of the first path of $h(t)$, respectively. Similarly, we denote $g_m(t)$ as the CIR with multi-path spread $\tau_{\sf g}$ and propagation delay $d_{\sf g}$ between the BD and the $m$-th receive antenna at the receiver, and $f_{m}(t)$ as the CIR with multi-path spread $\tau_{\sf f}$ and propagation delay $d_{\sf f}$ between the RF source and the $m$-th receive antenna at the receiver, respectively. \textcolor{black}{The corresponding passband channels are denoted as $\tilf_m(t)$'s, $\tilh(t)$, and $\tilg_m(t)$'s, respectively.} The backscatter channel is a concatenation of the source-to-BD channel and the BD-to-receiver channel. For convenience, we define the total channel spread for each channel as the sum of its multi-path spread and its channel propagation delay. \textcolor{black}{We assume that the relevant channels are independent from each other, which is a typical assumption in the literature such as~\cite{HeWangCL11}, \cite{YangBackscatter15}, and~\cite{BoyerRoyTWC13}}.

\subsubsection{Continuous-time Signal Model}
For the AmBC system, we first establish the model for continuous-time signals. Denote the passband signal transmitted from the RF source by
\begin{align}\label{eq:tx_signal_receiver_RF}
  \tils (t) \triangleq \text{Re}\left\{\sqrt{p} s(t) e^{j 2 \pi f_{\sf c} t}\right\},
\end{align}
where $s(t)$ is the baseband OFDM signal with unit power, $p$ is the average transmit power, and $f_{\sf c}$ represents the carrier frequency of the RF source.

Let $N$ be the number of subcarriers of the OFDM signal $s(t)$. In order to combat the inter-symbol-interference (ISI), a CP is added at the beginning of each OFDM symbol, and the CP length $t_{\sf c}$ is set to be longer than the maximum channel spread of all legacy OFDM receivers~\cite{GoldsmithWC2005}. \textcolor{black}{Typically, compared to the legacy OFDM receivers, the AmBC system is deployed in a place relatively close to the RF source, such that the BD can harvest more energy from the RF source. Based on this deployment criterion, we have two practical assumptions. First, we assume that the CP length $t_{\sf c}$ is much longer than the maximum channel spread of AmBC system. Second, we assume that the energy from the rechargeable battery at the BD is sufficient for its operation, and focus on the transceiver design for AmBC, for the purpose of exposition.}
%due to the space limitation.

From Fig.~\ref{fig:Fig1} and~\eqref{eq:tx_signal_receiver_RF}, the ambient OFDM signal received at the BD can be represented by using the baseband signal $s(t)$ and the baseband channel $h(t)$~\cite{GoldsmithWC2005}, i.e.,
\begin{align} \label{eq:rx_signal_BD_RF}
  \tilc(t) = \text{Re}\left\{ \left[\sqrt{p} s(t-d_{\sf h}) \otimes h(t) \right] e^{j 2 \pi f_{\sf c} (t-d_{\sf h})} \right\}.
\end{align}
Thus the corresponding baseband signal received at the BD is
\begin{align} \label{eq:rx_signal_BD}
  c(t) = \sqrt{p} s(t-d_{\sf h}) \otimes h(t).
\end{align}

Let $x(t)$ be the BD's baseband signal to be transmitted. Denote $\alpha$ as the complex attenuation (i.e., reflection coefficient) of the signal $\tilc(t)$ inside the BD, which depends on the antenna impedance and the load impedance of the BD. Then the backscattered signal out of the BD is $\alpha \tilc(t) x(t)$, where $\alpha$ controls the power of the backscattered signal by the BD.

\begin{myrem}
The above operation of backscattering can be viewed as a new modulation technique for wireless communication systems. The incident passband signal $\tilc(t)$, which is from the air, plays the role of carrier signal, and the BD signal $x(t)$ is the baseband modulation signal. We call the technique of modulating the BD signal $x(t)$ to the carrier $\tilc(t)$ as ``modulation in the air'', which exempts the BD from generating RF sinusoidal carriers locally and thus significantly reduces its hardware complexity and power consumption.
\end{myrem}

The received passband signal at the $m$-th antenna, $m=1, \ 2, \cdots \ M$, of the receiver is thus
\begin{align} \label{eq:rx_signal_receiver_RF}
  \tily_m(t) = \left[\alpha \tilc(t-d_{\sf g}) x(t-d_{\sf g})\right] \otimes \tilg_m(t)  +  \tils(t-d_{\sf f}) \otimes \tilf_m(t) + \tilw_m(t),
\end{align}
where $\tilw_m(t)$ is the passband noise at the receiver. After down-conversion to baseband, the received baseband signal at the receiver is written as
\begin{align} \label{eq:rx_signal_receiver}
  y_m(t) = y_{{\sf b}, m} (t) +  y_{{\sf d}, m} (t) + w_m(t),
\end{align}
where $y_{{\sf b}, m} (t) = [\alpha c(t-d_{\sf g}) x(t-d_{\sf g})] \otimes g_m(t)$ is the received backscattered signal from the BD,
$y_{{\sf d}, m} (t) = \sqrt{p} s(t-d_{\sf f}) \otimes f_m(t)$ is the direct-link interference from the RF source,
and $w_m(t)$ is the equivalent baseband additive white Gaussian noise (AWGN) with power $\sigma^2$, i.e., $w_m(t) \sim \calC \calN(0, \sigma^2)$. We assume that the noise term $w_m (t)$ is independent of the signals $y_{{\sf d}, m} (t)$ and $y_{{\sf b}, m} (t)$.

Due to short BD-to-receiver distance in practice, we assume that each channel $g_m(t)$ has a single path, denoted by $g_m$. The received backscattered signal at the $m$-th antenna of the receiver is thus simplified as %\thicksim
\begin{align}\label{eq:rx_signal_receiver2}
  y_{{\sf b}, m} (t) &= \alpha g_m c(t-d_{\sf g}) x(t-d_{\sf g}) = \alpha g_m \sqrt{p} s(t-d_{\sf h}-d_{\sf g}) \otimes h(t-d_{\sf g}) x(t-d_{\sf g}).
\end{align}

\begin{myrem}
  Since the AmBC is carried out at the same frequency band as the legacy OFDM system, the whole system in Fig.~\ref{fig:Fig1} can be considered as a spectrum sharing system~\cite{YCLiangTVT15, ZhangLiangJSTSP15}. The backscattered signal $\alpha \tilc(t) x(t)$ is also received by each nearby legacy OFDM user through the channel from the BD to the legacy user, resulting in interference to the legacy OFDM system. However, the power of this resulting interference is typically much lower than that of the received signal from the RF source, even if the legacy users and the BD are close to each other. \textcolor{black}{The reason is two-fold. First, the absolute value of reflection coefficient $\alpha$ is typically very small~\cite{Dobkinbook2007}. Second, the received signal from the RF source suffers from one-round channel attenuation, while the interference signal (i.e., backscattered signal) suffers from two-round channel attenuation, i.e., the source-to-BD channel and the BD-to-legacy-user channel.} In a similar argument, at the side of the receiver in the AmBC system, the direct-link interference from the RF source is also typically much stronger than the backscattered signal from the BD, which makes the BD signal detection a challenging task.
\end{myrem}

\subsubsection{Discrete-time Signal Model}
Let $f_{\sf s}$ be the sampling rate (or equivalently, the one-sided bandwidth) of the ambient OFDM signal. The discrete-time propagation delays for all channels are thus denoted as $D_{\sf h}=\lfloor d_{\sf h} f_{\sf s} \rfloor$, $D_{\sf g} =\lfloor d_{\sf g} f_{\sf s} \rfloor$ and $D_{\sf f} =\lfloor d_{\sf f} f_{\sf s} \rfloor$, respectively. Denote the propagation delay of the backscatter channel by $D_{\sf b} =\lfloor (d_{\sf h}+d_{\sf g}) f_{\sf s} \rfloor$. Define the minimum channel propagation delay as $D = \min\{D_{\sf f}, D_{\sf b}\}$. Since $d_{\sf h} \gg d_{\sf g}$ in practice, we assume that $D_{\sf b} \approx D_{\sf h}$. Similarly, the discrete-time total channel spreads are denoted as $L_{\sf h}=\lfloor (d_{\sf h}+\tau_{\sf h}) f_{\sf s} \rfloor$, $L_{\sf g} =\lfloor (d_{\sf g}+\tau_{\sf g}) f_{\sf s} \rfloor$ and $L_{\sf f} =\lfloor (d_{\sf f}+\tau_{\sf f}) f_{\sf s} \rfloor$, respectively. Denote the total channel spread of the backscatter channel by $L_{\sf b} =\lfloor (d_{\sf h}+\tau_{\sf h}+d_{\sf g}+\tau_{\sf g}) f_{\sf s} \rfloor$. Define the maximum channel spread as $L = \max\{L_{\sf f}, L_{\sf b}\}$.

For convenience, we rewrite the discrete-time source signal,  received signal at the BD, and the BD signal as $s[n]=s\left(n/f_{\sf s}\right)$, $c[n]=c\left(n/f_{\sf s}\right)$, and $x[n]=x\left(n/f_{\sf s}\right)$, respectively. After analog-to-digital conversion, the discrete-time representation of the received signal at the receiver is
\begin{align} \label{eq:rx_signal_receiver_dt}
  y_m[n] = y_{{\sf b}, m} [n] +  y_{{\sf d}, m} [n] + w_m[n],
\end{align}
where $y_{{\sf b}, m} [n] = \alpha c[n-D_{\sf g}] x[n-D_{\sf g}] \otimes g_m[n]$, $y_{{\sf d}, m} [n] = \sqrt{p} s[n-D_{\sf f}] \otimes f_m[n]$, and $w_m[n] \sim \calC \calN(0, \sigma^2)$. The relevant discrete-time signals are illustrated in Fig.~\ref{fig:Fig2} in Section~\ref{sec:OptDetection}.

For the studied AmBC system, our objective is to design the BD waveform and receiver detector to recover the BD signal $x[n]$ from the received signals $y_m[n]$'s at the receiver, without knowing the ambient signal $\sqrt{p} s[n]$ transmitted from the RF source. In the rest of this paper, we use the baseband signals and the baseband CIRs, for clarity of description. Before designing the transceiver, we propose a protocol for the AmBC system to operate in practice.

\subsection{Link-layer Protocol Design}~\label{sec:protocol}
As illustrated in Fig.~\ref{fig:Fig2B}, the proposed protocol adopts frame-based transmission, where each BD frame with (discrete) time duration $T_{\sf f}$ consists of four phases, the wake-up-preamble transmission (WUPT) phase with time duration $T_{\sf w}$, the blind timing synchronization (BTS) phase with time duration $T_{\sf b}$, the training-preamble transmission (TPT) phase with time duration $T_{\sf t}$, and the device data transmission (DDT) phase with time duration $T_{\sf d}$. The BD by default is in a sleep mode to save energy if it has no data to transmit. Once it has enough data to transmit, the BD is activated. In the following, we specify the operations in each phase.

\begin{figure}[t!]
\centering
\includegraphics[width=.85\columnwidth]{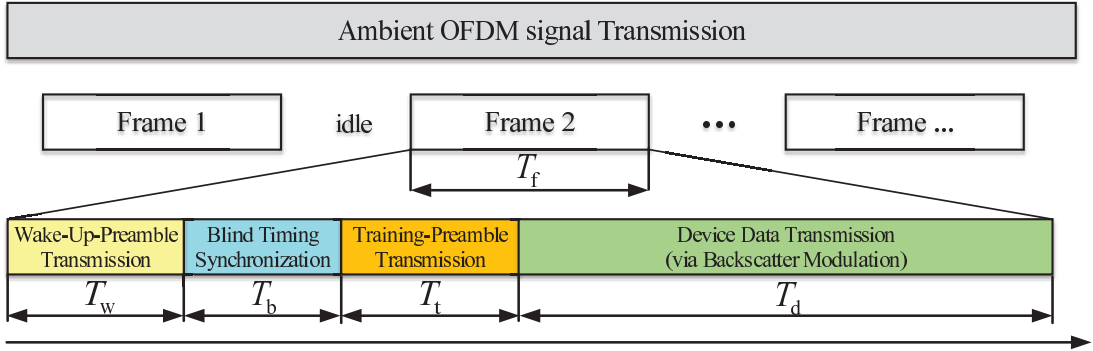}
\caption{Protocol design for AmBC system.}%from the BD
\label{fig:Fig2B}
\vspace{-0.8cm}
\end{figure}

\begin{itemize}
%\subsection{Wake-up-Preamble Transmission Phase}
  \item In the first WUPT phase, the BD is switched into the backscatter transmitter, and the BD sends a specialized preamble to activate the receiver's hardware. As in the literature~\cite{ABCSigcom13},~\cite{WiFiBackscatter14},~\cite{InterscatterSigcom16}, a short sequence of alternating `1' and `0' can be used as the wake-up preamble. The power consumption of the receiver can be saved via this event-driven wake-up scheme.
   \item In the second BTS phase, the BD switches its antenna into the information receiver and estimates the channel propagation delay $D_{\sf h}$ from the RF source. The BD performs blind timing estimation by exploiting the CP structure of $s[n]$, since it does not know the ambient OFDM signal $s[n]$. The blind estimation algorithm will be presented in Section~\ref{sec:timing_ABD}.
  \item In the third TPT phase, the BD switches its antenna into the backscatter transmitter, and chooses the estimate of $D_{\sf h}$ as the starting time to send a training preamble known by the receiver. The receiver receives both the direct-link signal $y_{\sf d}[n]$ and the preamble signal backscattered from the BD. Then the receiver estimates essential parameters including the minimum propagation delay $D$ and the maximum channel spread $L$, as well as the average signal power $\sigma_{u}^2$ when the BD is backscattering. Such estimation algorithms will be presented in Section~\ref{sec:timing_receiver}.
  \item In the fourth DDT phase, the BD switches its antenna into the backscatter transmitter, and transmits data bits to the receiver. Both the BD waveform and the receiver detector will be studied in Section~\ref{sec:OptDetection} in detail.
\end{itemize}

The time allocations for BTS and TPT phases affect both the parameter estimation accuracy and the communication throughput in the DDT phase. The effect on parameter estimation will be numerically investigated in Section~\ref{sec:simulation}. However, this paper focuses on the transceiver design, and will not derive the optimal timing allocation for all phases, due to the space limitation.

\section{Transceiver Design for Single-antenna System}\label{sec:OptDetection}
In this section, we establish a spread-spectrum (SS) model for the general signal model in~\eqref{eq:rx_signal_receiver_dt}, based on which we study the transceiver design, including the BD waveform design and the optimal receiver detector, for an AmBC system with a single-antenna receiver, i.e. $M=1$.

\subsection{A Spread-Spectrum Perspective for Signal Model}\label{subsec:spread_spectrum}
Since the switching frequency of backscatter state at the BD is typically much smaller than the OFDM sampling rate $f_{\sf s}$~\cite{Dobkinbook2007}, the signal backscattered by the BD can be viewed as the multiplication of a low-rate BD data signal $x[n]$ scaled by the reflection coefficient $\alpha$, and the high-rate spreading-code signal $c[n]$ in an over-the-air manner. The chip duration of this spreading code is equal to the sampling period of the received OFDM signal $c[n]$, which is much shorter than the duration of each BD symbol. Suppose that the BD symbol duration is designed to be equal to $N_0$ OFDM sampling periods (i.e., $N_{0}/f_{\sf s}$), the processing gain (or spreading factor), denoted by $G$, is then $G = N_0$, and the data rate of BD transmission is $\frac{f_s}{N_0}$. Hence, the BD symbol period needs to be designed carefully to achieve the optimal trade-off between the processing gain and the data rate.

Different from traditional SS communication systems, there are three main challenges for the receiver to detect the BD signal $x[n]$ from its received signals $y_m[n]$'s, which are listed as follows:
\begin{itemize}
  \item First, the spreading code $c[n]$ depends on the unknown ambient signal $\sqrt{p} s[n]$ and the fading channel $h[n]$, thus is time-varying, and unknown to the receiver. Therefore, the traditional correlation detection cannot be applied to the considered system.%AmBC
      %it is difficult to detect $x[n]$ by decorrelating $y[n]$ with $a[n]$.
  \item Second, the received direct-link interference signals $y_{{\sf d}, m} [n]$'s are typically unknown by the receiver and much stronger than the received backscattered signals $y_{{\sf b}, m} [n]$'s, resulting in very low signal-to-interference-noise ratio (SINR) for the receiver if they are treated as part of the background noise.
  \item Third, due to unknown ambient signal $\sqrt{p} s[n]$ and strong direct-link interferences $y_{{\sf d}, m} [n]$'s, it is challenging for the receiver to estimate all the fading channels $h[n]$, $g_m$'s and $f_m[n]$'s. Hence, coherent detection cannot be applied to the considered system.
\end{itemize}

To solve the above challenging problems, in the rest of this section, we focus on the joint transceiver design for the AmBC system with a single-antenna receiver. For notational simplicity, the subscript $m=1$ is omitted in the rest of this section.

\subsection{BD Waveform Design}\label{subsec: TX_design}
Specifically, we can design the time duration of each BD symbol to be equal to $K$ ($K \geq 1$) OFDM symbol periods each of which consists of $(N+N_{\sf c})$ sampling periods, i.e., the spreading gain is $G= K (N+N_{\sf c})$. As shown in Fig.~\ref{fig:Fig2_BD}, the BD uses the following waveform $x[n]$ to convey information bit $B=1$ in each BD symbol,
\begin{align}\label{eq:waveform_B1}
  x[n] = \left\{\begin{array}{cl}
  1, \; &\text{for} \quad n=(k-1)(N+N_{\sf c}),\; \ldots, \; (k-1)(N+N_{\sf c})+\frac{N+N_{\sf c}}{2}-1,\\
  -1, \; &\text{for} \quad n = (k-1)(N+N_{\sf c})+\frac{N+N_{\sf c}}{2}, \; \ldots, k(N+N_{\sf c})-1, \\
  \end{array}
  \right.
\end{align}
for $k=1, \; \ldots, \; K$, where for convenience we assume that $(N+N_{\sf c})$ is an even integer. The following waveform $x[n]$ is then used to convey information bit $B=0$ in each BD symbol,
\begin{align}\label{eq:waveform_B0}
  x[n] = 1, \; &\text{for} \quad n=0, \ldots, K(N+N_{\sf c})-1.
\end{align}
That is, for bit `1', there is a state transition in the middle of each OFDM symbol period within one BD symbol period, while for bit `0', there is no such transition.

\begin{figure}[!t]
\centering
\includegraphics[width=.8\columnwidth]{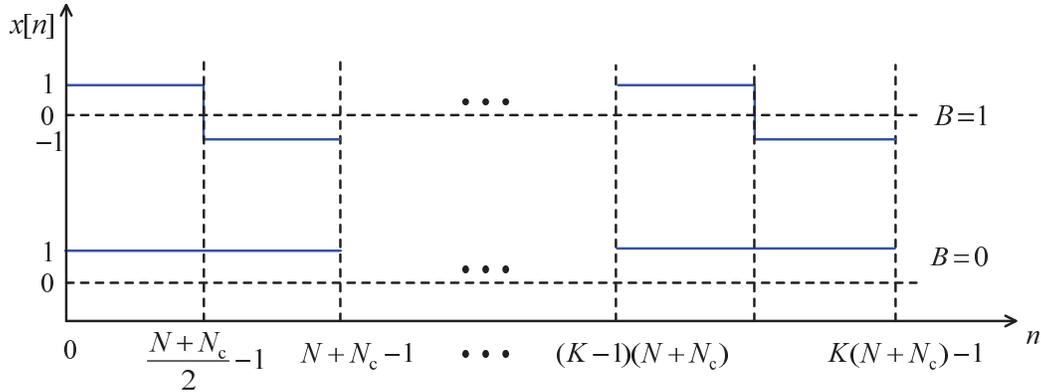}
\caption{BD waveform design.}%from the BD
\label{fig:Fig2_BD}
\vspace{-0.8cm}
\end{figure}

\begin{myrem}
The waveform design in~\eqref{eq:waveform_B1} and ~\eqref{eq:waveform_B0} aims to enable the receiver to cancel out the strong direct-link interference, as will be shown in the next subsection. \textcolor{black}{Also the designed waveform can be easily implemented in simple and low-cost BDs, since it is similar to FM0 waveform widely used in commercial RFID tags~\cite{Dobkinbook2007}.}
\end{myrem}

\subsection{Receiver Detection Design}\label{subsec:opt_detector}
In this subsection, we study the detector design at the receiver. For convenience, we choose $K=1$, without loss of generality.
\subsubsection{Construction of Test Statistic}\label{subsec:test}
As characterized in~\eqref{eq:rx_signal_receiver_dt} for $m=1$, the direct-link signal $y_{{\sf d}} [n]$ and the backscatter-link signal $y_{{\sf b}} [n]$ experience different multi-paths $f_m[n]$ and $h[n] g_m$, respectively. However, both $y_{{\sf d}} [n]$ and $y_{{\sf b}} [n]$ have its repeating structure, which are illustrated in Fig.~\ref{fig:Fig2}, since \textcolor{black}{the CP is inserted at the beginning of each OFDM symbol of the RF source signal $s[n]$}.

\begin{figure}[!t]
\centering
\includegraphics[width=.8\columnwidth]{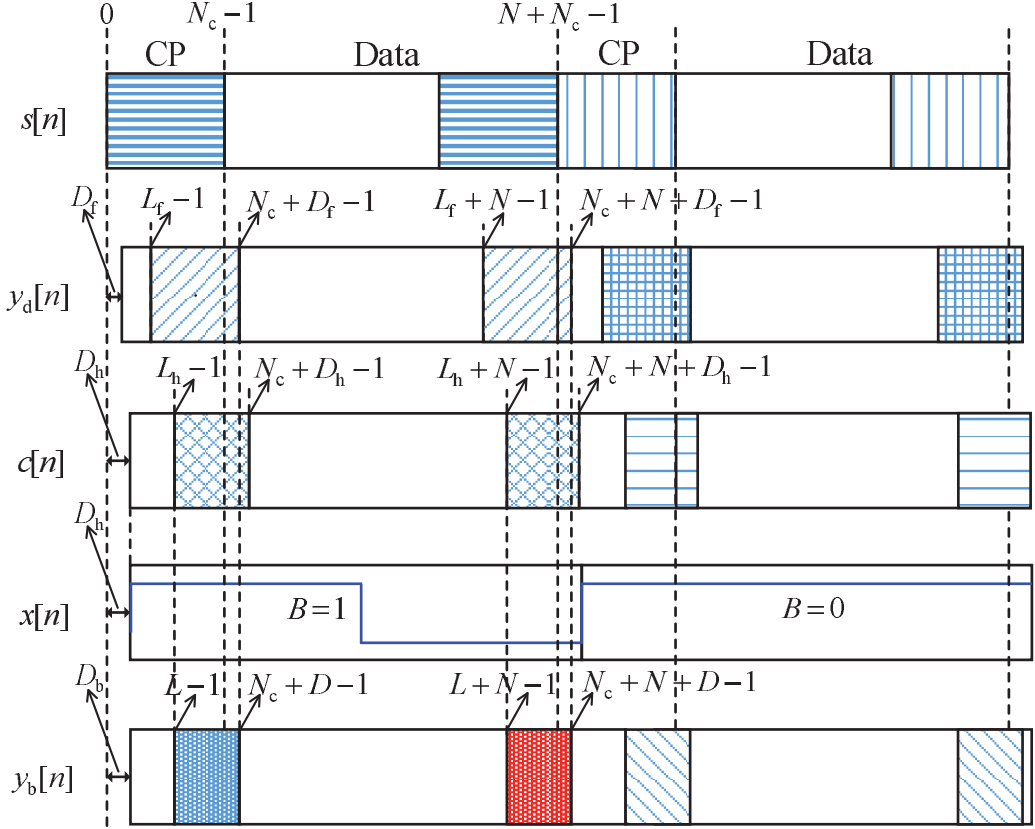} %{symbol_illustration_V170404.eps}%symbol_illustration_V0804
\caption{Signal structure for the case of $D_{\sf f} < D_{\sf h}$ and $L_{\sf f} < L_{\sf h}$.}%from the BD
\label{fig:Fig2}
\vspace{-0.8cm}
\end{figure}

To be specific, as illustrated in Fig.~\ref{fig:Fig2}, due to the use of CP and the multi-path effect, two portions of the received signal from the RF source, $y_{\sf d}[n]$, in each OFDM symbol period at the receiver are identical, i.e., \begin{align}\label{eq:ydn_repeat}
  y_{\sf d} [n]=y_{\sf d} [n+N], \; \; n=L_{\sf f}-1, \; \cdots, \; N_{\sf c}+D_{\sf f}-1.
\end{align}

Similarly, the repeating structure holds for the received signal at the BD, $c[n]$, i.e.,
\begin{align}\label{eq:cn_repeat}
  c[n]&= c[n+N], \; \; n=L_{\sf h}-1, \cdots, \; N_{\sf c}+D_{\sf h}-1.
\end{align}%\text{for}

From~\eqref{eq:waveform_B1},~\eqref{eq:waveform_B0}, and~\eqref{eq:cn_repeat}, the received backscatter-link signal $y_{\sf b} [n]$ at the receiver has the following repeating structure,
\begin{align}\label{eq:ybn_repeat}
  y_{\sf b} [n]
  &=\left\{ \begin{array}{cl}
  y_{\sf b} [n+N], \; &\text{if} \; B=0,  \\
  -y_{\sf b} [n+N], \; &\text{if} \; B=1.  \\
  \end{array}
  \right.
\end{align}
for $n=L_{\sf b}-1, \cdots, \; N_{\sf c}+D_{\sf b}-1$.

\textcolor{black}{By utilizing the repeating structure of $y_{\sf d}[n]$  in~\eqref{eq:ydn_repeat} and $y_{\sf b}[n]$  in~\eqref{eq:ybn_repeat}, we further have}
\begin{align}\label{eq:zn}
  z[n] \!\triangleq \! y[n] \!-\! y[n \!+\! N]
  \!&=\! \left\{ \begin{array}{cl}
  \! \! v[n], \! \! &\text{if} \; B=0,  \\
  \! \! u[n] \!+\! v[n], \! \! &\text{if} \; B=1,  \\
  \end{array}
  \right.
\end{align}
for $n=L-1, \cdots, \; N_{\sf c}+D-1$, where the signal $u[n]$ and the noise $v[n]$ are given as follows, respectively,
\begin{align}
  u[n] &= 2 \alpha g \sqrt{p} \sum \limits_{l=0}^{L_{\sf h}-1} s[n-l] h[l], \\
  v[n] &=w[n]-w[n+N].
\end{align}
Clearly, the noise $v[n] \sim \calC \calN(0, \sigma_v^2)$ with power\footnote{The receiver can estimate $\sigma_v^2$ offline, since $\sigma_v^2$ depends on only the noise variance.} $\sigma_v^2=2 \sigma^2.$

For large $N$, the received OFDM signal $c[n]$ at the BD is a sequence of \textcolor{black}{independent and identically distributed (i.i.d.)} random variables each of which follows the CSCG distribution with zero mean and variance $p \sum \nolimits_{l=0}^{L_{\sf h}-1} \left| h[l] \right|^2$~\cite{YCLiangTradeoffTWC08, GoldsmithWC2005}. Hence, the signal $u[n] = 2 \alpha g c[n]$ is a sequence of independent random variables each of which is identically distributed as $u[n] \sim \calC \calN (0, \sigma_u^2)$, where the variance $\sigma_u^2$ is given by %, i.e., $\tils[n] \sim \calC \calN(0, p)$.
\begin{align}
\sigma_u^2 = 4 p |\alpha|^2 |g|^2 \sum \limits_{l=0}^{L_{\sf h}-1} \left| h[l] \right|^2.
\end{align}

For notational simplicity, we define the detection SNR as
\begin{align}\label{eq:SNR}
  \gamma \triangleq \frac{\sigma_u^2}{\sigma_v^2} = \frac{2 p |\alpha|^2 |g|^2 \sum \nolimits_{l=0}^{L_{\sf h}-1} \left| h[l] \right|^2}{\sigma^2}.
\end{align}
Notice that there is a trade-off between the detection SNR $\gamma$ in~\eqref{eq:SNR} and the power available for harvesting at the BD. When the reflection coefficient $\alpha$ increases, the detection SNR $\gamma$ increases, but the available power for harvesting decreases~\cite{BoyerSumit14}, due to energy conservation law, and vice versa.

\begin{myrem}
We have two observations for the constructed intermediate signal $z[n]$ in~\eqref{eq:zn}. First, in $z[n]$, the direct-link interference $y_{\sf d} [n]$ is completely cancelled out, and only the received backscattered signal $y_{\sf b} [n]$ remains if $B=1$. This leads to higher detection SNR $\gamma$, thus tackles the challenge of strong direct-link interference at the receiver. Second, the statistic of the signal $u[n]$ (i.e., the signal power $\sigma_u^2$) hinges on only the strength of the overall backscatter channel $\alpha g h[n]$'s, independent of the phase information of $g$ and $h[n]$'s. This exempts the receiver from estimating the individual channels $g$ and $h[n]$'s, thus tackles the challenge of channel estimation at the receiver.
\end{myrem}

For convenience of analysis, we define the repeating length of the constructed signal $z[n]$ in~\eqref{eq:zn} as $J \triangleq N_{\sf c} +D - L+1$. Since both the signal $u[n]$ and the noise $v[n]$ are CSCG, the optimal detector is the energy detector~\cite{SMKayStatisticalSP93}. Therefore, we construct the following test statistic,
\begin{align}\label{eq:testR}
  R&= \frac{1}{J \sigma_v^2} \sum \limits_{n=L-1}^{N_{\sf c}+D-1} \left| z[n] \right|^2.
\end{align}

\textcolor{black}{Clearly, for general case of arbitrary $J \geq 1$, the exact distribution of $R$ is Chi-square distribution with degrees of freedom $2J$. When the number of summation terms in~\eqref{eq:testR} (i.e., $J$) is large\footnote{Note that in practice $J$ can be effectively made arbitrarily large by increasing the spreading gain parameter, $K$, even with finite $N_{\sf c}$, $L$ and $D$ values.}, the central limit theorem (CLT)~\cite{WalpoleProbbook2011} implies that the distribution of $R$ can be approximated by the Gaussian distribution, which is given by the following lemma.}
%Moreover, we have the following lemma on the conditional distribution of $R$.
%For large $N_{\sf c}$ such that (N_{\sf c} - L+1)
\begin{mylem} \label{lem1}
When the repeating length $J$ is large, the conditional distribution of the test statistic $R$ is given by
\begin{align}\label{eq:cond_PDF}
  R = \left\{ \begin{array}{cl}
  R|_{B=0}  \sim \calN \left( \mu_0, \sigma_0^2\right), \quad &\text{if} \; B=0,  \\
  R|_{B=1}  \sim \calN \left( \mu_1, \sigma_1^2\right), \quad &\text{if} \; B=1, \\
  \end{array}
  \right.
\end{align}
where the mean values are
\begin{align}
  \mu_0 = 1, \quad \mu_1 = \gamma+1, \label{eq:mu_01}
\end{align}
and the variance values are
\begin{align}
  \sigma_0^2 = \frac{1}{J}, \quad \sigma_1^2 = \frac{(\gamma+1)^2}{J}. \label{eq:sigma2_01}
\end{align}
\end{mylem}

\begin{proof}
  The proof is based on the CLT, and follows similar steps as in~\cite{YCLiangTradeoffTWC08}. Please refer to Appendix~\ref{App:Lemma1}.
\end{proof}
% for details
%law of large number (LLN) central limit theorem (CLT)~\cite{WalpoleProbbook2011},

%, the test statistic $R$ follows chi-square distribution with degree of freedom $2J$
%\textcolor{black}{For general case of arbitrary $J \geq 1$, the exact distribution of $R$ conditioned on BD bit $B$ is Chi-square distribution with degree of freedom $2J$. The value of $R$ is always zero. However, for the case of a large repeating length $J$, the conditional distribution of $R$ can be approximated by the Gaussian distribution given in Lemma~\ref{lem1}. }

%Notice that  gives the approximate distribution of $R$ conditioned on the BD bit $B$, for ,

\subsubsection{Optimal Detector Design}\label{subsubsec:opt_detector}
Let $p(R |_{B=0})$ and $p(R |_{B=1})$ be the probability density function (PDF) of the conditional random variable $R |_{B=0}$ and $R |_{B=1}$, respectively. Since $B=0$ and $B=1$ are equally probable, the optimal detector follows the ML rule\cite{SMKayStatisticalSP93}, i.e., %maximum likelihood (
\begin{align}\label{eq:ML}
  \hatB = \left\{ \begin{array}{cl}
  0, \quad &\text{if} \; p(R |_{B=0}) > p(R |_{B=1}),  \\
  1, \quad &\text{if} \; p(R |_{B=1}) > p(R |_{B=0}). \\
  \end{array}
  \right.
\end{align}

In other words, the decision rule is $\hatB=0$ if $R <\epsilon$, and $\hatB=1$ otherwise, where $\epsilon$ is the detection threshold. From~\eqref{eq:cond_PDF}, the BER is %given by Given a threshold $\epsilon>0$,
\begin{align}\label{eq:BER}
  P_{\sf e} (\epsilon) &= \frac{1}{2} P(\hatB=1 | B=0) + \frac{1}{2} P(\hatB=0 | B=1) = \frac{1}{2} P_{\sf fa} (\epsilon)  + \frac{1}{2} P_{\sf md} (\epsilon),
%  &= \frac{1}{2} Q\left( \frac{\epsilon - \mu_0}{\sqrt{\sigma_0^2}}\right) + \frac{1}{2} Q\left( \frac{\mu_1 - \epsilon}{\sqrt{\sigma_1^2}}\right),
\end{align}
where the probability of false alarm $P_{\sf fa} (\epsilon)$ and the probability of missing detection $P_{\sf md} (\epsilon)$ are given as follows, respectively,
\begin{align}
  P_{\sf fa} (\epsilon) &= \calQ \left( \frac{\epsilon - \mu_0}{\sqrt{\sigma_0^2}}\right) = \calQ\left( \sqrt{J} (\epsilon - 1) \right), \nonumber \\
  P_{\sf md} (\epsilon) &= \calQ \left( \frac{\mu_1 - \epsilon}{\sqrt{\sigma_1^2}}\right) = \calQ \left( \sqrt{J} \left(1- \frac{\epsilon}{\gamma + 1}\right) \right).
\end{align}
%where the $\calQ$-function $\calQ(x) = \frac{1}{\sqrt{2 \pi}} \int_x^{\infty} e^{-r^2/2} dr$.
The BER is illustrated as the shadow area in Fig.~\ref{fig:Fig3}.
\begin{figure}[!t]
\vspace{-0.5cm}
\centering
\includegraphics[width=.65\columnwidth] {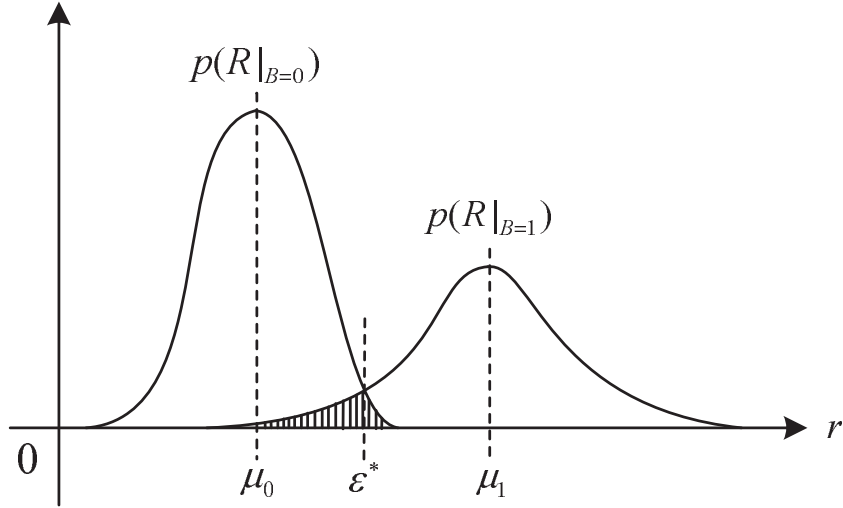}
\caption{Two conditional PDFs and the corresponding BER region (shadow area).}% of false detection of $p(R |_{B=1})$ and $p(R |_{B=0})$,
\label{fig:Fig3}
\vspace{-0.5cm}
\end{figure}

Hence, the optimal threshold $\epsilon^{\star}$ is the solution to the following problem
\begin{subequations}\label{eq:opt1}
\begin{align}
   \textrm{(P1)} \ \ \ \underset{\epsilon}{\min} \ \
    &P_{\sf{e}} (\epsilon) = \frac{1}{2} \calQ\left( \sqrt{J} (\epsilon - 1) \right) + \frac{1}{2} \calQ \left( \sqrt{J} \left(1- \frac{\epsilon}{\gamma + 1}\right) \right) \label{eq:BER_def} \\
    \quad \text{s. t.} \ \
        &\epsilon > 0.
\end{align}
\end{subequations}

The optimal threshold $\epsilon^{\star}$ and the corresponding minimum BER $P_{\sf{e}, \min}$ are given in the following theorem.
\begin{mythe}\label{thm1}
The optimal detection threshold for the ML detector in~\eqref{eq:ML} is given by
\begin{align}\label{eq:opt_threshold}
  \epsilon^{\star} = \frac{\gamma+1}{\gamma (\gamma+2)} \left( \gamma + \sqrt{\gamma^2 + \frac{2 \gamma (\gamma+2) \ln (\gamma+1)}{J}}\right).
\end{align}
And the corresponding minimum BER is given by
\begin{align}\label{eq:min_BER}
  P_{\sf{e}, \min} &= \frac{1}{2} \calQ\left( \sqrt{J} (\epsilon^{\star} - 1) \right) + \frac{1}{2} \calQ \left( \sqrt{J} \left(1- \frac{\epsilon^{\star}}{\gamma + 1}\right) \right).
\end{align}
\end{mythe}

\begin{proof}
Please refer to Appendix~\ref{app:theorem_ML_BER_SA}.
\end{proof}

%Given the parameter $J$, both the optimal detection threshold $\epsilon^{\star}$ and the corresponding minimum BER depend on only the detection SNR $\gamma$.

\begin{myrem} \label{rem:Thms}
From Theorem~\ref{thm1}, we have two observations. First, the optimal detection threshold $\epsilon^{\star}$ decreases as the repeating length $J$ increases. Also, $\epsilon^{\star} \rightarrow \frac{2(\gamma +1)}{\gamma+2}$, as $J \rightarrow \infty$. It can be easily checked that the limiting threshold $\epsilon^{\star}$ is located in the closed interval between the two conditional mean values $\mu_0$ and $\mu_1$, due to $\gamma \geq 0$. Second, from Theorem~\ref{thm1}, we observe that given $J$, the optimal detection threshold and the corresponding minimum BER depend on only the detection SNR $\gamma$, being irrelevant to the number of subcarriers $N$. This is because only the CP-induced repeating parts of the received signal is utilized to detect the BD information bits, without using other parts in each OFDM symbol period.
\end{myrem}

\begin{myrem}\label{rem:Difference}
\textcolor{black}{Although the derivation of BER and the detection threshold follows similar steps as standard binary detection~\cite{JGProakisMSalehi05}, our system model, BD signal design and test statistics construction are fundamentally different from existing literature on BER analysis for AmBC~\cite{QianGaoAmBCTWC16, WangGaoAmBCTCOM16}.}
\end{myrem}

\section{Parameter Estimation for AmBC Systems}\label{sec:parameter_est}
%Time Synchronization
In this section, we present practical methods for estimating essential parameters for implementing AmBC systems. In practice, the BD and the receiver can estimate basic parameters from the received OFDM signals. For instance, they can estimate $N$ as the distance of two adjacent peaks of the autocorrelation of the received signal.

\subsection{Timing Estimation for BD}\label{sec:timing_ABD}
In the BTS phase, the BD is switched into information receiver mode and estimates the channel propagation delay $D_{\sf h}$, which is required for the proposed modulation at the BD as in~\eqref{eq:waveform_B1} and~\eqref{eq:waveform_B0}. For the case in which the BD knows the synchronization preamble in the ambient OFDM signals\footnote{This is practical in some scenarios. For instance, for WLAN systems with 802.11a standard~\cite{80211a}, fixed frequency-domain OFDM symbols generate a synchronization preamble which consists of several identical training symbols.}, it can estimate $D_{\sf h}$ accurately by using traditional estimation methods such as cross-correlation based method or other algorithms reviewed in~\cite{MorelliKuoProceed07}.
%(e.g., cross-correlation based method and other algorithms in~\cite{MorelliKuoProceed07} therein).

For the case in which the BD does not know the synchronization preamble, it can still use its received incident signal $c[n]$ to estimate $D_{\sf h}$ by utilizing the repeating CP structure. Specifically, the BD performs autocorrelation for the received signals within the time window of $K_1$ ($K_1 \geq1$) OFDM symbol periods, i.e., $T_{\sf b}=K_1 (N+N_{\sf c})$, and estimates $D_{\sf h}$ as
\begin{align}\label{eq:Dh_est}
  \widehat{D_{\sf h}}  = \frac{1}{K_1} \underset{d=0, \ldots, N_{\sf c}-1}{\arg \max} \sum \limits_{k=0}^{K_1 - 1} \sum \limits_{n=0}^{N_{\sf c}-1} \frac{\big| c\left[n+d+k(N+N_{\sf c})\right] c^{\ast}\left[n+d+N+k(N+N_{\sf c})\right]\big|}{\Big|c\big[n+d+N+k(N+N_{\sf c})\big]\Big|^2}.
\end{align}

\subsection{Parameter Estimation for Receiver}\label{sec:timing_receiver}
We first consider timing synchronization for the case in which the receiver knows the synchronization preamble in the ambient OFDM signals. The receiver can use cross-correlation based method to estimate the direct-link propagation delay $D_{\sf f}$ and the backscatter-link propagation delay $D_{\sf b}$ accurately, in the BTS phase and the TPT phase, respectively. Then it can obtain an estimate of $D=\min \{D_{\sf f}, D_{\sf b} \}$. While in the TPT phase, it also obtains the estimated maximum channel spread $L$ by estimating the maximum multi-path delay in the frequency domain~\cite{MorelliKuoProceed07}, in which the backscatter link is treated as additional multi-path.

Second, we consider timing synchronization for the case in which the receiver does not know the synchronization preamble. In the TPT phase, the BD is switched into backscattering mode and the receiver estimates the minimum channel propagation delay $D$ and the maximum channel spread $L$. For convenience, we set $T_{\sf t}=K_2 (N+N_{\sf c})$, and choose the training preamble sent by the BD as $x_{\sf t}[n]=1, \; \text{for} \quad n=0, \ldots, T_{\sf t}-1$. The receiver can estimate $D$ by autocorrelating the received signal $y[n]$, which is similar to the estimation of $D_{\sf h}$ as in~\eqref{eq:Dh_est} and thus omitted herein.

Utilizing the repeating structure of $y_{\sf d}[n]$  in~\eqref{eq:ydn_repeat} and $y_{\sf b}[n]$  in~\eqref{eq:ybn_repeat}, we construct the metric %following
% for estimating $L_{\sf b}$
\begin{align}\label{eq:Lb_est_metrc}
  Q[l] = \frac{1}{K_2(N_{\sf c} - l)}\sum \limits_{k=0}^{K_2 - 1} \sum \limits_{n=0}^{N_{\sf c}-1-l} \Big| y[n+l+k(N+N_{\sf c})] - y[n+l+N+k(N+N_{\sf c})] \Big|^2.
\end{align}

For typical case of $N_{\sf c} \gg L$, $Q[l]$ is distributed as $Q[l] \sim \calC \calN \left( \widetilde{\mu}_u(l), \widetilde{\sigma}_u^2 (l) \right)$, with the mean value
\begin{align}
  \widetilde{\mu}_u^2 (l) =
  \left\{ \begin{array}{cl}
  \frac{2\sigma^2 }{N_{\sf c} - l} \left[ (L_{\sf f}-l) (\gamma_{\sf d}+\gamma) + (L_{\sf b}-L_{\sf f}) \gamma + N_{\sf c}-L_{\sf b} \right], \quad &\text{if} \; l < L_{\sf f} < L_{\sf b}\\
  \frac{2\sigma^2 }{N_{\sf c} - l} \left[ (L_{\sf b}-l) (\gamma_{\sf d}+\gamma) + (L_{\sf f}-L_{\sf b}) \gamma_{\sf d} + N_{\sf c}-L_{\sf f} \right], \quad &\text{if} \; l <  L_{\sf b} < L_{\sf f}\\
  \frac{2\sigma^2 }{N_{\sf c} - l} \left[(L_{\sf b}-l) \gamma  + N_{\sf c} - L_{\sf b} \right] , \quad &\text{if} \;  L_{\sf f} < l <  L_{\sf b} \\
  \frac{2\sigma^2 }{N_{\sf c} - l} \left[(L_{\sf f}-l) \gamma_{\sf d}  + N_{\sf c} - L_{\sf f} \right] , \quad &\text{if} \;  L_{\sf b} < l <  L_{\sf f} \\
  2\sigma^2, \quad &\text{if} \; l \geq L. \\
    \end{array}
  \right.%_{\sf b}
\end{align}%_{\sf max}
and the variance value
\begin{align}
  \widetilde{\sigma}_u^2 (l) =
  \left\{ \begin{array}{cl}
  \frac{4\sigma^4}{K_2 (N_{\sf c} - l)^2} \left[ (L_{\sf f} - l) (\gamma_{\sf d}+\gamma+1)^2 + (L_{\sf b} - L_{\sf f}) (\gamma+1)^2  + N_{\sf c}-L_{\sf b}\right], \quad &\text{if} \; l < L_{\sf f} < L_{\sf b}\\
  \frac{4\sigma^4}{K_2 (N_{\sf c} - l)^2} \left[ (L_{\sf b} - l) (\gamma_{\sf d}+\gamma+1)^2 + (L_{\sf f} - L_{\sf b}) (\gamma_{\sf d}+1)^2  + N_{\sf c}-L_{\sf f}\right], \quad &\text{if} \; l <  L_{\sf b} < L_{\sf f}\\
  \frac{4\sigma^4}{K_2 (N_{\sf c} - l)^2} \left[ (L_{\sf b} - l) (\gamma+1)^2 + N_{\sf c}-L_{\sf b}\right], \quad &\text{if} \;  L_{\sf f} < l <  L_{\sf b} \\
  \frac{4\sigma^4}{K_2 (N_{\sf c} - l)^2} \left[ (L_{\sf f} - l) (\gamma_{\sf d}+1)^2 + N_{\sf c}-L_{\sf f}\right], \quad &\text{if} \;  L_{\sf b} < l <  L_{\sf f} \\
  \frac{4\sigma^4}{K_2 (N_{\sf c} - l)}, \quad &\text{if} \; l \geq L.
    \end{array}
  \right.%_{\sf b}
\end{align}

Clearly, when $l \geq L$, the metric $Q[l]$ is almost zero due to pure noise remaining in the difference signal; otherwise, $Q[l]$ is relatively large due to signal components. Therefore, the maximum channel spread $L$ can be estimated by using Algorithm~\ref{Algorithm_L}, shown in the above table.

\begin{algorithm}[t!]
	\caption{: Algorithm for estimating $L$} \label{Algorithm_L}
\begin{algorithmic}[1]
\STATE Initialization: some positive constant $\epsilon.$ \\
\FOR{$l=0,\ldots,N_{\sf c}-1$}
\STATE Compute $Q[l]$ in~\eqref{eq:Lb_est_metrc}.
\IF{$Q [l] \leq 2 \epsilon \sigma^2$, }
\STATE Obtain $\widehat{L} = l$.
\ENDIF
\ENDFOR
\RETURN $\widehat{L}$.
\end{algorithmic}
\end{algorithm}

Third, we consider the estimation of the average signal power $\sigma_{u}^2$ when the BD is backscattering. Using the estimated $\widehat{L}$, the parameter $\sigma_{u}^2$ is estimated as
\begin{align}
  \widehat{\sigma}_{u}^2=\frac{1}{N_{\sf c}} \sum \limits_{n=0}^{N_{\sf c}-1} \left| y[n+\widehat{L}] - y[n+\widehat{L}+N] \right|^2.
\end{align}

\section{Transceiver Design for Multi-Antenna Receiver}\label{sec:MA_extension}
In this section, we study the transceiver design for an AmBC system with a receiver equipped with $M$ antennas, for $M>1$. We use the same BD waveform (i.e., transmitter) design as that for the single-antenna system studied in Section~\ref{sec:OptDetection}, to maintain the ability of direct-link interference cancellation and information decoding by leveraging the repeating structure of the relevant signals due to CP. We thus focus on the optimal receiver design for the case of multi-antenna receiver in this section. The objective of the multi-antenna receiver is to detect the BD information bits, by processing the signals received at $M$ antennas jointly and then making the final decision based on the calculated statistic. In the following, the test statistic is first constructed, and then the optimal detector is obtained.

\subsection{Construction of Test Statistic}\label{sec:MA_test_statistics}
Similar to Section~\ref{subsec: TX_design}, we construct the following intermediate signal based on the signal received by the $m$-th antenna, $m=1, \cdots, M$, at the receiver,
\begin{align}\label{eq:zn_MA}
  z_m[n] \!\triangleq \! y_m[n] \!-\! y_m(n \!+\! N)
  \!&=\! \left\{ \begin{array}{cl}
  \! \! v_m[n], \! \! &\text{if} \; B=0,  \\
  \! \! u_m[n] \!+\! v_m[n], \! \! &\text{if} \; B=1,  \\
  \end{array}
  \right.
\end{align}
where the signal $u_m[n] = 2 \alpha g_m \sqrt{p} \sum \nolimits_{l=0}^{L_{\sf h}-1} s(n-l) h(l)$ and the noise $v_m[n] =w_m[n]-w_m(n+N)$. The signal $u_m[n]$ follows CSCG distribution with zero mean and variance given by
 \begin{align}
  \sigma_{u,m}^2 = 4 p |\alpha|^2 |g_m|^2 \sum \limits_{l=0}^{L_{\sf h}-1} \left| h[l] \right|^2.
 \end{align}
And the noise follows that $v_m[n] \sim \calC \calN(0, \sigma_v^2)$. We define the detection SNR of the $m$-th antenna as
\begin{align}
  \gamma_m \triangleq \frac{\sigma_{u,m}^2} {\sigma_v^2}= \frac{2 p |\alpha|^2 |g_m|^2 \sum \nolimits_{l=0}^{L_{\sf h}-1} \left| h[l] \right|^2}{\sigma^2}.
\end{align}

We then construct the following test statistic for each receiver antenna $m$,
\begin{align}
  R_m&= \frac{1}{J \sigma_v^2} \sum \limits_{n=L-1}^{N_{\sf c}+D-1} \left| z_m [n] \right|^2.
\end{align}

We further construct the following test statistic for the final decision,
\begin{align}
  \tilR&= \sum \limits_{m=1}^M \theta_m R_m,
\end{align}
where the combination weights $\theta_m$'s are subject to $\sum \limits_{m=1}^M \theta_m^2 =1$ and $\theta_m \geq 0, \ \forall m$. Denote $\bm{\theta}=[\theta_1 \ \theta_2 \ \cdots \ \theta_M]^T$.

We assume that the channels for different antennas at the receiver are mutually independent. Similar to Lemma~\ref{lem1}, when the repeating length $J$ is large, the conditional distribution of the test statistic $\tilR$ is given by
\begin{align}\label{eq:cond_PDF_MA}
  \tilR = \left\{ \begin{array}{cl}
  \tilR|_{B=0} \sim \calN \left( \widetilde{\mu}_0, \widetilde{\sigma}_0^2\right), \quad &\text{if} \; B=0,  \\
  \tilR|_{B=1} \sim \calN \left( \widetilde{\mu}_1, \widetilde{\sigma}_1^2\right), \quad &\text{if} \; B=1, \\
  \end{array}
  \right.
\end{align}
where the values of mean are
\begin{align}
  \widetilde{\mu}_0 = \widetilde{\mu}_0 (\bm{\theta}) = \sum \limits_{m=1}^M \theta_m, \quad \widetilde{\mu}_1 =\widetilde{\mu}_1 (\bm{\theta}) = \sum \limits_{m=1}^M \theta_m (\gamma_m+1), \label{eq:mu_01_MA}
\end{align}
and the values of variance are
\begin{align}
  \widetilde{\sigma}_0^2 =\widetilde{\sigma}_0^2(\bm{\theta}) = \frac{1}{J}, \quad \widetilde{\sigma}_1^2 = \widetilde{\sigma}_1^2 (\bm{\theta}) = \frac{1}{J}\sum \limits_{m=1}^M \theta_m^2 (\gamma_m +1)^2. \label{eq:sigma2_01_MA}
\end{align}

\vspace{-0.8cm}
\subsection{Optimal Detector Design}\label{sec: MA_opt_detector}
%Similar to Section~\ref{sec:OptDetection},
For the optimal receiver design, the objective is to find the optimal combination weights $\bm{\theta}$ and detection threshold $\tilepsilon$, such that the following BER is minimized.
\begin{align}\label{eq:BER_MA}
  \tilP_{\sf e} (\bm{\theta}, \tilepsilon) &= \frac{1}{2} P_{\sf fa} (\bm{\theta}, \tilepsilon)  + \frac{1}{2} P_{\sf md} (\bm{\theta}, \tilepsilon) = \frac{1}{2} \calQ \left( \sqrt{J} \left(\tilepsilon - \tilmu_0(\bm{\theta})\right) \right)  + \frac{1}{2} \calQ \left( \frac{\tilmu_1(\bm{\theta}) - \tilepsilon}{\sqrt{\tilsigma_1^2(\bm{\theta})}}\right).
\end{align}
That is, the optimization problem can be formulated as follows
\begin{subequations}\label{eq:opt2}
\begin{align}
   \textrm{(P2)} \ \ \ \underset{\bm{\theta}, \ \tilepsilon}{\min} \ \
    &\tilP_{\sf{e}} (\bm{\theta}, \tilepsilon) = \frac{1}{2} \calQ \left( \sqrt{J} \left(\tilepsilon - \tilmu_0(\bm{\theta})\right) \right)  + \frac{1}{2} \calQ \left( \frac{\tilmu_1(\bm{\theta}) - \tilepsilon}{\sqrt{\tilsigma_1^2(\bm{\theta})}}\right) \\
    \quad \text{s. t.} \ \
            &\sum_{m=1}^M \theta_m^2 = 1,  \\
        &\theta_m \geq 0, \;\; m=1, \; \cdots \; M \\
        &\tilepsilon > 0.
\end{align}
\end{subequations}

It is difficult to solve $\textrm{(P2)}$ directly, due to the complicated interplay between $\bm{\theta}$ and $\tilepsilon$ in the $\calQ$-function. Before simplifying $\textrm{(P2)}$, we have the following proposition.
\begin{mypro}\label{pro:cond_threshold}
Given $\bm{\theta}$, the optimal threshold for minimizing BER is given by
\begin{align}\label{eq:opt_threshold_MA}
  \tilepsilon^{\star} (\bm{\theta}) = \frac{\tilmu_0 (\bm{\theta}) \tilsigma_1^2 (\bm{\theta}) - \tilmu_1 (\bm{\theta}) + \sqrt{ \tilsigma_1^2 (\bm{\theta}) (\tilmu_1 (\bm{\theta}) - \tilmu_0 (\bm{\theta}))^2 + \tilsigma_1^2 (\bm{\theta}) (\tilsigma_1^2 (\bm{\theta}) - 1) \log(\tilsigma_1^2 (\bm{\theta}))}}{\tilsigma_1^2 (\bm{\theta}) - 1}.
\end{align}
And the corresponding minimum BER is given by
\begin{align}\label{eq:min_BER_MA1}
  \tilP_{\sf{e}, \min} (\bm{\theta}) &= \frac{1}{2} \calQ\left( f_1(\bm{\theta}) \right) + \frac{1}{2} \calQ\left( f_2(\bm{\theta}) \right),
\end{align}
where the two positive-valued functions of $\bm{\theta}$ in the above are given by
\begin{align}
  f_1(\bm{\theta}) &= \sqrt{J} \left(\tilepsilon^{\star} (\bm{\theta}) - \tilmu_0 (\bm{\theta}) \right), \label{eq:f1}\\ % > 0
  f_2(\bm{\theta}) &= \frac{ \tilmu_1 (\bm{\theta}) - \tilepsilon^{\star} (\bm{\theta})}{\sqrt{\tilsigma_1^2 (\bm{\theta})}},  \label{eq:f2}%  >0
\end{align}
%and the two functions satisfy
with $\tilmu_0 (\bm{\theta}), \tilmu_1 (\bm{\theta})$ and $\tilsigma_1^2 (\bm{\theta})$ given in~\eqref{eq:mu_01_MA} and~\eqref{eq:sigma2_01_MA}, respectively.
\end{mypro}
\begin{proof}
  It is noted that given $\bm{\theta}$, the BER is minimized when the threshold $\tilepsilon (\bm{\theta})$ is chosen as the intersection point of the two conditional PDFs $p(\tilR |_{B=0}, \bm{\theta})$ and $p(\tilR |_{B=1}, \bm{\theta})$. Following the same steps as the proof of Theorem~\ref{thm1}, this proposition is proved.
\end{proof}
%It can be further checked that $f_1(\bm{\theta}) \geq 0$ and $f_2(\bm{\theta}) \geq 0$.
From the condition $p(\tilR |_{B=0}, \bm{\theta}) = p(\tilR |_{B=1}, \bm{\theta})$, we further have
\begin{align}\label{eq:f12}
  f_1^2(\bm{\theta})=f_2^2(\bm{\theta}) + \ln\left(\tilsigma_1^2 (\bm{\theta})\right).
\end{align}

From Proposition~\ref{pro:cond_threshold}, the optimal combination weights $\bm{\theta}$ are chosen to minimize the conditional BER $\tilP_{\sf{e}, \min} (\bm{\theta})$ in~\eqref{eq:min_BER_MA1}. From~\eqref{eq:f12}, Problem $\textrm{(P2)}$ is equivalent to the following optimization problem,
\begin{subequations}\label{eq:opt2}
\begin{align}
   \textrm{(P2-Equi.)}\ \ \ \underset{\bm{\theta}}{\min} \ \
    &\tilP_{\sf{e}, \min} (\bm{\theta}) = \frac{1}{2} \calQ\left( \sqrt{f_2^2(\bm{\theta}) + \ln\left(\tilsigma_1^2 (\bm{\theta})\right)} \right) + \frac{1}{2} \calQ\left( f_2(\bm{\theta}) \right) \label{eq:opt2_obj}\\
    \quad \text{s. t.} \ \
        &\sum_{m=1}^M \theta_m^2 = 1,  \\
        &\theta_m \geq 0, \;\; m=1, \; \cdots \; M.
\end{align}
\end{subequations}
The objective function of $\textrm{(P2-Equi.)}$ is complicated, and the optimal $\bm{\theta}^{\star}$ can be obtained by $(M-1)$-dimensional search. \textcolor{black}{Since the number of receive antennas $M$ is small or moderate in practice, due to limited size of receiver, the complexity of $(M-1)$-dimensional search is acceptable.}

In Section~\ref{sec:simulation}, we will numerically compare the BER performance of the optimal combining weights $\bm{\theta}^{\star}$, to those of three traditional combining schemes including the maximum-ratio-combining (MRC), equal-gain-combining (EGC), and selection-combining (SC). The traditional schemes determine the combination weights based on only the SNRs of individual receive branches, thus are simpler to obtain. Numerical results will show that compared to the scheme using the optimal combination weights $\bm{\theta}^{\star}$, the MRC, EGC and SC schemes suffer from negligible SNR losses in terms of BER performance, thus are good suboptimal and low-complexity combining schemes in practice.

\section{Performance Analysis}\label{sec:detection_analysis}
In this section, we analyze the rate performance and the BER performance for the proposed transceiver design. % in Subsection~\ref{subsec:rate} and Subsection~\ref{subsec:BER_SA}, respectively,

\subsection{Rate Performance}\label{subsec:rate}
Recall the designed BD waveform in~\eqref{eq:waveform_B1} and~\eqref{eq:waveform_B0}. Since the time duration of each BD symbol is equal to $K (N+N_{\sf c})$ sampling periods, the BD rate is obtained as
\begin{align}\label{eq:rate}
  R_{\sf BD} = \frac{f_s}{K(N+N_{\sf c})}.
\end{align}

We observe that there is a trade-off between the data rate and the reliability, for different choices of BD symbol periods. For larger $K$, the BD data rate is lower, but the detection reliability at the receiver improves, since more signal samples are available for detection decision; and vice versa.

\subsection{BER Performance for Single-receive-antenna System}\label{subsec:BER_SA}
We analyze the effect of the CP length $N_{\sf c}$ and the channel spread $L$ on the BER performance for the single-antenna receiver case as follows.

From Theorem~\ref{thm1}, the minimum BER $P_{\sf{e}, \min}$ is rewritten as
\begin{align}\label{eq:min_BER1}
  P_{\sf{e}, \min} (J, \gamma) &= \frac{1}{2} \calQ\left( f_1 (J, \gamma) \right) + \frac{1}{2} \calQ\left( f_2 (J, \gamma) \right),
\end{align}
with the two quantities
\begin{align}
  f_1 (J, \gamma) &= \frac{\sqrt{J}}{\gamma (\gamma+2)} \left( (\gamma+1) \sqrt{\gamma^2 + \frac{2 \gamma (\gamma+2) \ln (\gamma+1)}{J}} - \gamma \right), \\
  f_2 (J, \gamma) &= \frac{\sqrt{J}}{\gamma (\gamma+2)} \left( \gamma^2 + \gamma - \sqrt{\gamma^2 + \frac{2 \gamma (\gamma+2) \ln (\gamma+1)}{J}} \right).
\end{align}
It can be further checked that $f_1 (J, \gamma) >0$ and $f_2 (J, \gamma) >0$.

From the condition $p(R |_{B=0}) = p(R |_{B=1})$, we have
\begin{align}\label{eq:f12_SA}
  f_1^2 (J, \gamma) =f_2^2 (J, \gamma) + 2 \ln\left(\gamma + 1 \right) > f_2^2 (J, \gamma).
\end{align}
The minimum BER $P_{\sf{e}, \min}$ is thus rewritten as
\begin{align}\label{eq:min_BER2}
  P_{\sf{e}, \min} (J, \gamma) &= \frac{1}{2} \calQ\left( \sqrt{f_2^2 (J, \gamma) + 2 \ln\left(\gamma + 1 \right)} \right) + \frac{1}{2} \calQ\left( f_2 (J, \gamma) \right).
\end{align}

Moreover, we have the following proposition on the effect of $N_{\sf c}$ and $L$ on the BER performance.
\begin{mypro}\label{pro:BER_D}
Given $\gamma$, the minimum BER $P_{\sf{e}, \min}$ decreases, as either $N_{\sf c}$ increases or $L$ decreases.
% and the optimal threshold $\epsilon^{\star}$ given in Theorem~\ref{thm1}
\end{mypro}
\begin{proof}
  It can be checked that given $\gamma$, $f_2(J, \gamma)$ is an increasing function of $J$. The minimum BER $P_{\sf{e}, \min}$ decrease as $J$ increases, due to the fact that $\calQ$-function is a decreasing function of its argument. Since $J = N_{\sf c} +D - L + 1$, this proposition is proved.
\end{proof}
%Intuitively, given $\gamma$ and thus fixed conditional variances, when $J$ increases, the two conditional PDFs in Fig.~\ref{fig:Fig3} become far apart from each other due to larger $(\mu_1 - \mu_0)$, thus leading to smaller shadow area and equivalently lower BER. Proposition~\ref{pro:BER_D} verifies this intuition.

Next, we analyze the effect of the detection SNR $\gamma$ on the BER performance. We focus on the typical case that the (decimal) SNR $\gamma$ is sufficiently large such that $\frac{1}{\gamma} \approx 0$.

% \eqref{eq:mu_01}, \eqref{eq:sigma2_01} and
From~\eqref{eq:min_BER1}, the minimum probability of false alarm $P_{\sf fa, \min}$ is approximated as % given by
\begin{align}\label{eq:opt_P_fa_app}
  P_{\sf{fa}, min} (J, \gamma) \approx   \calQ \left( \sqrt{J + 2 \ln (\gamma + 1)}\right),
\end{align}
and the minimum probability of missing detection $P_{\sf md, \min}$ is approximated as % given by
\begin{align}%\label{eq:opt_P_md_app}
  P_{\sf{md}, min} (J, \gamma)
  %&\approx  \calQ \left( J - \sqrt{\frac{2 \ln (\gamma + 1)}{\gamma (\gamma +2)}} \right) \label{eq:opt_P_md_app1} \\
  &\approx  \calQ \left( \sqrt{J}\right), \label{eq:opt_P_md_app2}
\end{align}%\overset{(a)}{
where the approximation in~\eqref{eq:opt_P_md_app2} is from the inequality $\log(1+x) \leq x, \ \forall x > -1,$  and the assumption that $\frac{1}{\gamma} \approx 0$.

Given $J$, the minimum probability of false alarm $P_{\sf fa, \min}$ dominates the minimum BER $P_{\sf e, \min} = \frac{1}{2}P_{\sf fa, \min} + \frac{1}{2} P_{\sf md, \min}$, since the minimum probability of missing detection $P_{\sf md, \min}$ approximately equals the constant $\calQ \left( \sqrt{J} \right)$. Note $P_{\sf fa, \min}$ decreases as $\gamma$ increases, due to the fact that the $\calQ$-function is a monotonically decreasing function. Thus we directly have the following proposition.
 %on the BER performance.
\begin{mypro}\label{pro:gamma}
  Given $N_{\sf c}$ and $L$, for a decimal SNR $\gamma$ that is sufficiently large such that $\frac{1}{\gamma} \approx 0$, the BER $P_{\sf e, min}$ decreases as the SNR $\gamma$ increases.
\end{mypro}
Clearly, Proposition~\ref{pro:gamma} coincides with the intuition that the BER decreases as the SNR $\gamma$ increases.

\subsection{BER Performance for Multiple-receive-antenna System}\label{subsec:BER_MA}
%For the case without timing error,
With the optimal combination weights $\bm{\theta}^{\star}$, the minimum BER is obtained from~\eqref{eq:opt2_obj} as
\begin{align}\label{eq:min_BER1_MA}
\tilde{\tilP}_{\sf{e}, \min} (\bm{\theta}^{\star}) &= \frac{1}{2} \calQ\left( \sqrt{f_2^2(\bm{\theta}^{\star}) + \ln\left(\tilsigma_1^2 (\bm{\theta}^{\star})\right)} \right) + \frac{1}{2} \calQ\left( f_2(\bm{\theta}^{\star}) \right).
\end{align}

\section{Numerical Results}\label{sec:simulation}
In this section, we provide simulation results to evaluate the performance of the proposed transceiver design. Suppose that the OFDM signal bandwidth or sampling frequency $f_{\sf s}$ is 10MHz. The channel taps are modeled as statistically independent Gaussian random variables with zero mean (Rayleigh fading) and an exponentially decaying power delay profile. We set $\tau_{\sf g}=1$, and assume that the channel gain $\bbE[|g_m|^2]=\frac{c^2}{4 \pi D_{\sf rd}^2 f_{\sf c}^2}, \forall m=1, \ \ldots, \ M$, where the light speed $c=3 \times 10^8$ meters per second, the carrier frequency $f_{\sf c} = 900$MHz, and the distance between the BD and the receiver $D_{\sf rd}=0.5$ meter (m) which implies the channel propagation delay $D_{\sf g}=\, \lfloor \frac{D_{\sf rd} f_{\sf s}}{c}\rfloor=0$. We set the channel parameters $D_{\sf f}=16, D_{\sf h}=16$, and $\tau_{\sf f}=4, \tau_{\sf h}=6$, which implies $D=16$ and $L=22$. \textcolor{black}{We set the reflection coefficient $\alpha=0.3 + 0.4 j$, which implies that about 25\% incident power is reflected by the BD.} We assume that the BD adopts binary phase shift keying (BPSK) modulation. We set the number of subcarriers $N=512$ and the CP length $N_{\sf c}=64$. The following numerical results are based on $10^8$ Monte Carlo simulations each with randomly generated channels.

\subsection{Timing Estimation for AmBC System}\label{sec:MSE_TE}
In this subsection, we evaluate the performance of timing estimation for AmBC systems. We set the parameter $\epsilon=1.5$ for Algorithm~\ref{Algorithm_L}.

Fig.~\ref{fig:MSE_timing} plots the normalized mean-square-error (MSE) versus the SNR $\gamma$, for different synchronization time $K_1$'s or $K_2$'s. For estimating $D_{\sf h}$ by the BD using the conventional autocorrelation based method, the MSE is about 0.01, since the SNR at the BD is about 20 dB larger than the SNR $\gamma$ at the receiver due to the small reflection coefficient $\alpha$ and the channel attenuation from the BD to the receiver. For autocorrelation based estimation of $D$ at the receiver, the MSE decreases quickly as the SNR $\gamma$ increases. We observe that the MSE curve has a floor for the SNR $\gamma > 15$dB, which verifies that the performance of autocorrelation based synchronization for OFDM signals is dominated by the delay profile and not sensitive to the noise level~\cite{JGProakisMSalehi05}. For estimating $L$ at the receiver, the MSE also decreases as the SNR $\gamma$ increases. In particular, we observe that for the SNR $\gamma=30$ dB, the normalized MSE is $0.016$, $0.01$, and $0.008$ for $K_2 =1$, $2$, and $3$, respectively.
%%In particular, we observe that the normalized MSE is $5 \times 10^{-3}$, $3.1 \times 10^{-3}$, and $2.8 \times 10^{-3}$ for $K=1$, $2$, and $3$, respectively. $1.6 \times 10^{-2}$, $1.0 \times 10^{-2}$, and $8 \times 10^{-3}$

\begin{figure}[t!]
\centering
\includegraphics[width=.65\columnwidth] {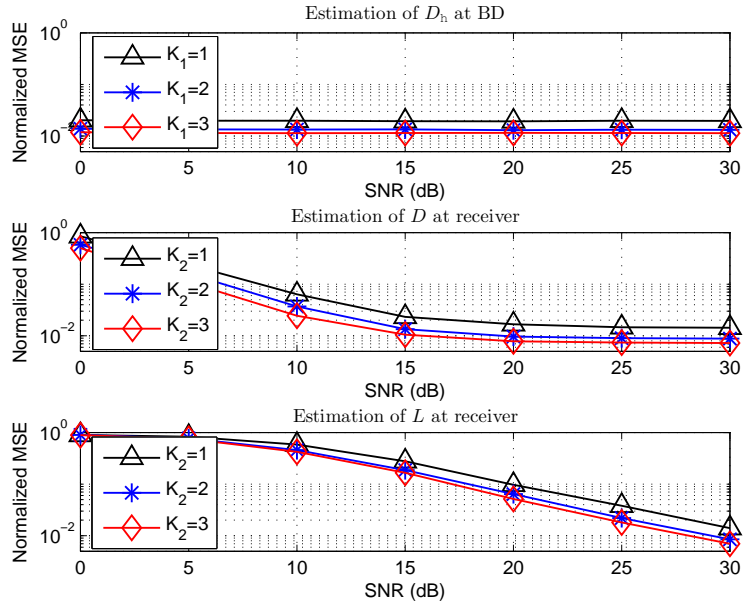}
\caption{Normalized MSE comparison for estimating $D_{\sf h}$, $D$ and $L$.}
	\label{fig:MSE_timing} %Dh_Lb_Est_V170224
\vspace{-0.8cm}
\end{figure}

For estimating $D_{\sf h}$, $D$ and $L$, we observe that the normalized MSE becomes smaller for longer synchronization time (i.e., larger $K_1$ or $K_2$), which implies that it is sufficient to use one or two OFDM symbols for timing synchronization in practice. Since we have $N_{\sf c} \gg D_{\sf h}$ and $N_{\sf c} \gg L$ in practice, the MSE performance achieved by the proposed method is sufficient for implementation. Thus we assume perfect timing synchronization at the BD and the receiver in the subsequent simulations.
%we observe that for the SNR $\gamma=25$ dB, the normalized MSE is $4 \times 10^{-2}$, $2.5 \times 10^{-2}$, and $2 \times 10^{-2}$ for $K=1$, $2$, and $3$, _{\sf b} _{\sf b}

% at the BD and the receiver, respectively

\subsection{Performance Comparison for Case of Single-antenna Receiver}
In this subsection, we evaluate the BER performance for the case of a single-antenna receiver. For performance comparison, we consider the energy detector in~\cite{ABCSigcom13} as a benchmark, in which, for the case of $K=1$,  the BD reflects for bit `1' and keeps silent for bit `0', and the receiver detects the BD bit by distinguishing between two different energy levels of the received signal $y[n]$, given by
\begin{align}
  \hatR \triangleq \frac{1}{N+N_{\sf c}} \sum \limits_{n=0}^{N+N_{\sf c}-1} |y[n]|^2.
\end{align}
Differential coding is used in~\cite{ABCSigcom13} to exempt the receiver from knowing the extra mapping from the power levels to the bits.
\begin{figure}[t!]
\centering
\includegraphics[width=.65\columnwidth] {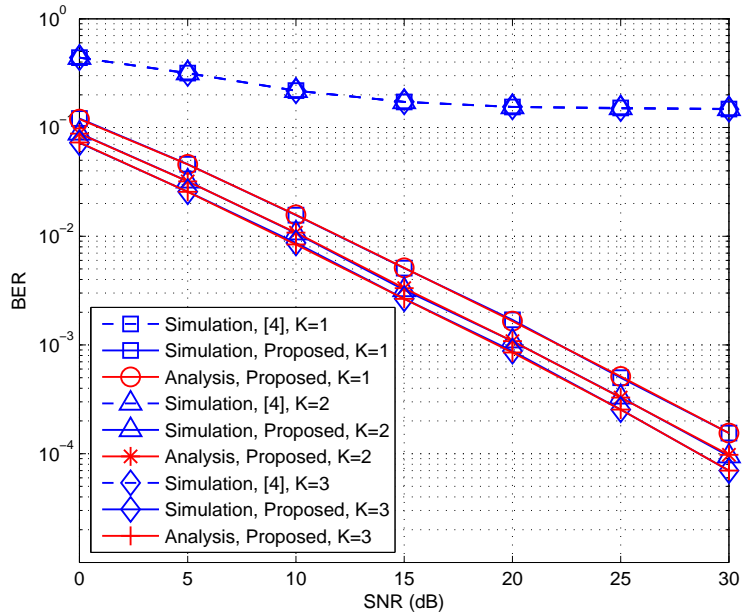}
\caption{BER comparison for both proposed and conventional designs.}
\label{fig:Fig4}%BER_SA_V170301A.eps, for $N=512, \ N_{\sf c}=64$
\vspace{-0.6cm}
%\vspace{-0.3cm}%f1_ber_0406_N2048_Ncp256 f_0510_N512_CP64_DeltaT0_compare.eps timing_0318_sumNcpppDhA
\end{figure}

\subsubsection{Scenario of Fixed Distance between BD and Receiver}% Number of Subcarriers $N$
%In this subsection, w
We fix the distance between the BD and the receiver as $D_{\sf rd}=0.5$ m.

Fig.~\ref{fig:Fig4} plots the BER versus the average SNR, by using the proposed optimal detector and the conventional energy detector, for different BD symbol duration $K$'s. We observe that by using the proposed optimal detector, the BER for $K=1$ decreases dramatically from $0.12$ to $1.6 \times 10^{-4}$, as the average SNR increases from 0 dB to 30 dB. Moreover, simulation results verify the trade-off between the BD rate and the reliability of signal recovery. With the parameter setting in our simulations, the BD rate is 19.5 Kbps, 9.8 Kbps and 6.5 Kbps, for $K=1, \ 2$ and 3, respectively. We observe that the BER decreases for larger $K$. Specifically, for a BER level of $0.001$, the system achieves an SNR gain of $2$ dB and $3$ dB for $K=2$ and 3, respectively, compared to the case of $K=1$. Also, we observe that the simulated BERs coincide with the analytical BERs, which verifies Theorem~\ref{thm1}.

In contrast, by using the conventional energy detector, the BER decreases slowly, saturating at a high BER around $0.16$. This is explained as follows. The energy detector decodes the BD bit by treating the strong direct-link interference from the RF source as noise. Since the direct-link interference is typically much stronger than the backscattered signal, this results in very low decoding SNR and thus high BER floor.

\subsubsection{Scenario of Varying Distance between BD and Receiver}% Number of Subcarriers $N$
In this example, we vary the distance $D_{\sf rd}$ between the BD and the receiver. We set $K=1$.

Fig.~\ref{fig:Fig4A} plots the BER versus the distance $D_{\sf rd}$, for both the proposed transceiver and the conventional energy detector~\cite{ABCSigcom13}. In general, the BER increases as $D_{\sf rd}$ increases. For the proposed transceiver design, the BER is around $0.001$, $0.01$ and $0.1$, for the distance $D_{\sf rd}=1.4, \ 4$ and 14 meters. While for the benchmark scheme, the BER is around $0.15$, $0.4$ and $0.5$, for the distance $D_{\sf rd}=0.5, \ 2$ and 6 meters. It is concluded that the proposed design enhances the BER performance as well as the operating range significantly, compared to the conventional energy detector.

\begin{figure}[t!]
\centering
\includegraphics[width=.65\columnwidth] {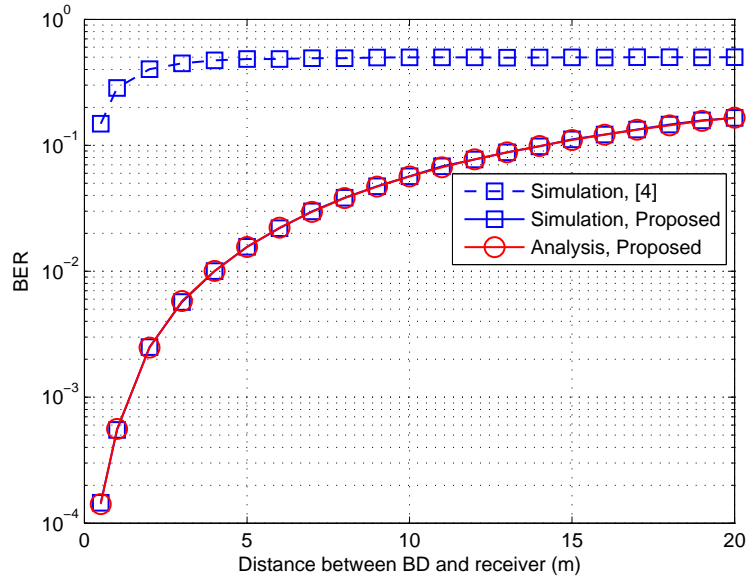}
\caption{BER comparison for different distances between BD and receiver.}
\label{fig:Fig4A}
\vspace{-0.8cm}
\end{figure}

\subsection{Performance Comparison for Case of Multi-antenna Receiver}
In this subsection, we evaluate the BER performance for the case of a multi-antenna receiver.
\begin{figure}[t!]
\centering
\includegraphics[width=.65\columnwidth] {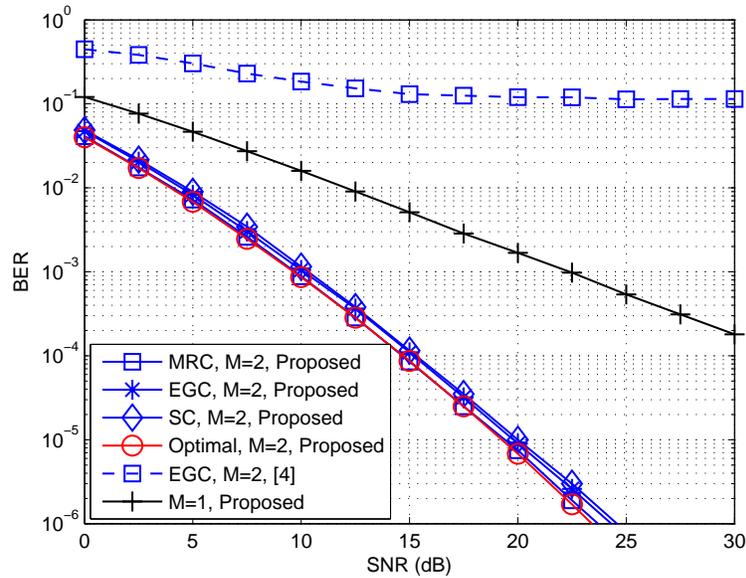}
\caption{BER comparison for different combining schemes, for $M=2$.}
\label{fig:Fig8}
\vspace{-0.8cm}
%\vspace{-0.3cm}%f1_ber_0406_N2048_DiffNcp f_0516_N512_CP64_M2_compare_DeltaT0
\end{figure}
Fig.~\ref{fig:Fig8} compares the BER of different combining schemes, for the number of receiver antennas $M=2$. The optimal combining weights are obtained by one-dimensional search, and the search-step is set to be 0.001. First, we observe that by using two antennas at the receiver, the proposed transceiver achieves an SNR gain of about 12 dB at a BER level of 0.001, compared to the case of a single-antenna receiver. This verifies that the receive diversity can decrease the BER significantly. Second, the BER performance difference for different combining schemes is small. In particular, at a BER level of $0.001$, compared to the optimal combining scheme, the MRC, EGC and SC schemes suffer from an SNR loss of 0.2 dB, 0.5 dB, and 0.6 dB, respectively. For the conventional energy detector, the BER decreases very slowly with SNR, saturating at a high BER around $0.11$.

\begin{figure}[!t]
\vspace{-0.1cm}
\centering
\includegraphics[width=.65\columnwidth]{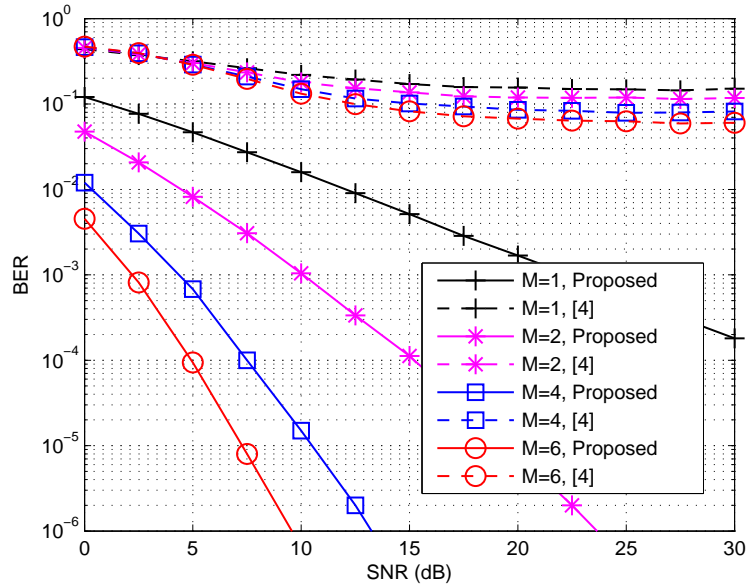}
\caption{BER comparison for different $M$'s, with EGC.}
\label{fig:Fig9}
\vspace{-1cm}
\end{figure}

The MRC and SC schemes require the SNR information of signals received at each antenna, which requires additional estimation. The above observations imply that the simple EGC scheme is suitable for the scenarios with unknown SNRs, since the EGC scheme does not require the SNR information. The EGC scheme can also reduce the computational complexity, as it avoids the search for optimal combination weights.

Fig.~\ref{fig:Fig9} compares the BER with different number of receiver antennas $M$'s, for the EGC scheme. First, we observe that the BER decreases quickly as $M$ increases. In particular, at a BER level of $0.001$, the system achieves an SNR gain of 12 dB, 18 dB, and 20 dB, for $M=$2, 4, 6, respectively, compared to the single-antenna case. This implies that the incremental SNR gain becomes smaller as $M$ increases. \textcolor{black}{We also observe that the BER improvement becomes less significant as $M$ increases. This is in accordance with the BER performance of receiver diversity via EGC~\cite{GoldsmithWC2005}.}  For the conventional energy detector~\cite{ABCSigcom13}, the BER decreases slowly as $M$ increases, due to the strong direct-link interference.

\section{Conclusions}\label{conslusion}
This paper has studied a new backscatter communication system over ambient OFDM carriers in the air. We first establish the system model for such system from a spread-spectrum modulation perspective, upon which a novel joint design for BD waveform and receiver detector is proposed. For the system with a single-antenna receiver, we construct a test statistic that cancels out the direct-link interference by exploiting the repeating structure of the relevant signals due to the use of CP, and propose the optimal maximum-likelihood detector to recover the BD information bits, for which the optimal detection threshold is obtained in closed-form expression. For the system with a multi-antenna receiver, we further construct a new test statistic and derive the corresponding optimal detector. To perform optimal detection, the receiver requires to estimate only the strength of the backscatter channel, instead of the complete information of the relevant channels, leading to reduced receiver complexity. We also propose efficient algorithms for timing synchronization in the considered AmBC system. The effect of various system parameters on the transmission rate and detection performance is analyzed. Simulation results have shown that the proposed design outperforms the conventional design based on energy detection, in terms of transmission rate, BER performance and operating range. Also, the results have shown that the proposed timing synchronization method is practically valid and efficient, and the deployment of multiple receive antennas at the receiver can enhance the BER performance significantly. The proposed transceiver design has great potential for applications in the next-generation low-power IoT systems.

%and RFID systems.
%\newpage
\appendices
\section{Proof to Lemma \ref{lem1}} \label{App:Lemma1}
%the law of large number (
Under the condition of $B=1$, from CLT\cite{WalpoleProbbook2011}, we have that $R|_{B=1}  \sim \calN \left( \mu_1, \sigma_1^2\right)$. The mean value is first given by
  \begin{align}\label{eq:mu1_proof}
    \mu_1   &\triangleq \frac{1}{J \sigma_v^2} \sum \limits_{n=L-1}^{N_{\sf c}-1} \bbE \left[ \left| u[n] + v[n] \right|^2 \right] \nonumber \\
    &\eqa \frac{\bbE [|u[n]+v[n]|^2]}{\sigma_v^2} \nonumber \\
    &\eqb \frac{\bbE [|u[n]|^2] + \bbE [|v[n]|^2]}{\sigma_v^2} \nonumber \\
    &\eqc \gamma +1,
  \end{align}
where $(a)$ is from the fact that random variables $u[n]$'s are i.i.d. and $v[n]$'s are also i.i.d., $(b)$ is from the mutual independence between $u[n]$ and $v[n]$ for any $n$, and $(c)$ is from the facts that $\bbE \left[ |u[n]|^2 \right] = \gamma \sigma_v^2$ and $\bbE \left[ |v[n]|^2 \right] = \sigma_v^2$.

Next, the variance value is given by
  \begin{align}
    \sigma_1^2 &\triangleq \frac{1}{J^2 \sigma_v^4} \bbE \left[\left| \sum \limits_{n=L-1}^{N_{\sf c}-1} \left( |u[n]+v[n]|^2 - \bbE \left[ |u[n]+v[n]|^2 \right] \right) \right|^2 \right] \nonumber \\
%    &\triangleq  \bbE \left[\left| \sum \limits_{n=L-1}^{N_{\sf c}-1} \left( \frac{|u[n]+v[n]|^2}{J \sigma_v^2} - \mu_1 \right) \right|^2 \right] \nonumber \\
    &\eqa \frac{1}{J^2 \sigma_v^4} \bbE \left[ \sum \limits_{n=L-1}^{N_{\sf c}-1} \left( |u[n]+v[n]|^2 - (\sigma_u^2 + \sigma_v^2) \right)^2 \right] \nonumber \\
    &\eqb \frac{\bbE \left[ \left( |u[n]+v[n]|^2 - (\sigma_u^2 + \sigma_v^2) \right)^2 \right]}{J \sigma_v^4} \nonumber \\
    &\eqc \frac{\bbE \left[ |u[n]|^4 \right] + \bbE \left[ |v[n]|^4 \right] - (\sigma_u^2 - \sigma_v^2)^2}{J \sigma_v^4} \nonumber \\
    &\eqd \frac{(\gamma+1)^2}{J},
  \end{align}
where $(a)$ is from~\eqref{eq:mu1_proof}, $(b)$ is from the fact that random variables $u[n]$'s are i.i.d. and $v[n]$'s are also i.i.d., $(c)$ is from the mutual independence between $u[n]$ and $v[n]$ for any $n$, and $(d)$ is from the facts that $\bbE \left[ |u[n]|^4 \right] = 2 \sigma_u^4,$ and $\bbE \left[ |v[n]|^4 \right] = 2 \sigma_v^4.$

Also, the distribution under the condition of $B=0$ can be proved in a similar way.

\section{Proof to Theorem~\ref{thm1}} \label{app:theorem_ML_BER_SA}
As can be seen from the shadow area in Fig.~\ref{fig:Fig3}, the BER is minimized when the threshold $\epsilon$ is chosen as the intersection point of the two conditional PDFs $p(R |_{B=0})$ and $p(R |_{B=1})$. From~\eqref{eq:cond_PDF}, the optimal $\epsilon^{\star}$ is thus the solution to the following equation,
\begin{align}\label{eq:PDF_equation}
  \frac{1}{\sqrt{2 \pi \sigma_0^2}} \exp \left( -\frac{(t \!- \!\mu_0)^2}{2 \sigma_0^2}\right) =   \frac{1}{\sqrt{2 \pi \sigma_1^2}} \exp \left( -\frac{(t \!-\! \mu_1)^2}{2 \sigma_1^2}\right).
\end{align}
Define the constant $C \triangleq (\gamma+1)^2$. After taking the logarithm on both sides of \eqref{eq:PDF_equation} and some manipulations, the equation~\eqref{eq:PDF_equation} is simplified as
\begin{align}
  \frac{C \!-\! 1}{2} T^2 \!+\! (\mu_1 \!-\! C \mu_0) T + \frac{C \mu_0^2 \!-\! \mu_1^2 \!-\! \sigma_1^2 \ln C}{2} =0.
\end{align}

Solving the above equation yields the optimal threshold
\begin{align}\label{eq:opt_threshold1}
  \epsilon^{\star} = \frac{C \mu_0 - \mu_1 + \sqrt{C (\mu_1- \mu_0)^2 + (C-1) \sigma_1^2 \ln C}}{C-1}.
\end{align}
Substituting~\eqref{eq:mu_01} and~\eqref{eq:sigma2_01} in \eqref{eq:opt_threshold1}, the optimal detection threshold is obtained as in~\eqref{eq:opt_threshold}, and the corresponding minimum BER is given in~\eqref{eq:min_BER}.

%\newpage
\renewcommand{\baselinestretch}{1.3}
\bibliography{IEEEabrv,reference1704}

\begin{thebibliography}{10}
\providecommand{\url}[1]{#1}
\csname url@samestyle\endcsname
\providecommand{\newblock}{\relax}
\providecommand{\bibinfo}[2]{#2}
\providecommand{\BIBentrySTDinterwordspacing}{\spaceskip=0pt\relax}
\providecommand{\BIBentryALTinterwordstretchfactor}{4}
\providecommand{\BIBentryALTinterwordspacing}{\spaceskip=\fontdimen2\font plus
\BIBentryALTinterwordstretchfactor\fontdimen3\font minus
  \fontdimen4\font\relax}
\providecommand{\BIBforeignlanguage}[2]{{%
\expandafter\ifx\csname l@#1\endcsname\relax
\typeout{** WARNING: IEEEtran.bst: No hyphenation pattern has been}%
\typeout{** loaded for the language `#1'. Using the pattern for}%
\typeout{** the default language instead.}%
\else
\language=\csname l@#1\endcsname
\fi
#2}}
\providecommand{\BIBdecl}{\relax}
\BIBdecl

\bibitem{YangLiangGC16}
G.~Yang and Y.-C. Liang, ``Backscatter communications over ambient {OFDM}
  signals: transceiver design and performance analysis,'' in \emph{Proc. of
  {IEEE} Globecom}, Washington DC, USA, Dec. 2016, pp. 1--6.

\bibitem{BiZhangComMag15}
S.~Bi, C.~K. Ho, and R.~Zhang, ``Wireless powered communication: opportunities
  and challenges,'' \emph{{IEEE} Commun. Mag.}, vol.~53, no.~4, pp. 117--125,
  Apr. 2015.

\bibitem{BiZhangWirelessCom16}
S.~Bi, Y.~Zeng, and R.~Zhang, ``Wireless powered communication networks: an
  overview,'' \emph{{IEEE} Wireless Commun.}, vol.~23, no.~4, pp. 10--18, Apr.
  2016.

\bibitem{ABCSigcom13}
V.~Liu, A.~Parks, V.~Talla, S.~Gollakota, D.~Wetherall, and J.~R. Smith,
  ``Ambient backscatter: wireless communication out of thin air,'' in
  \emph{Proc. of ACM SIGCOMM}, Hong Kong, China, Jun. 2013, pp. 1--13.

\bibitem{YCLiangTVT15}
Y.-C. Liang, K.-C. Chen, G.~Y. Li, and P.~Mahonen, ``Cognitive radio networking
  and communications: An overview,'' \emph{{IEEE} Trans. Veh. Technol.},
  vol.~60, no.~7, pp. 3386--3407, Sept. 2011.

\bibitem{Dobkinbook2007}
D.~M. Dobkin, \emph{The {RF} in {RFID}: Passive {UHF RFID} in Practice}.\hskip
  1em plus 0.5em minus 0.4em\relax Elsevier, 2007.

\bibitem{SampleSmith07}
A.~Sample, D.~Yeager, P.~Powledge, and J.~Smith, ``Design of a
  passively-powered, programmable sensing platform for {UHF RFID} systems,'' in
  \emph{Proc. of {IEEE} Int. Conf. on {RFID}}, Grapevine, TX, Mar. 2007, pp.
  149--156.

\bibitem{ParksSmith14}
A.~Parks and J.~Smith, ``Shifting through the airwaves: efficient and scalable
  multiband {RF} harvesting,'' in \emph{Proc. of IEEE Conf. RFID}, Orlando, FL,
  USA, Apr. 2014, pp. 74--81.

\bibitem{Ishizaki11}
H.~Ishizaki, H.~Ikeda, Y.~Yoshida, T.~Maeda, T.~Kuroda, and M.~Mizuno, ``A
  battery-less {WiFi-BER} modulated data transmitter with ambient radio-wave
  energy harvesting,'' in \emph{IEEE Symp. VLSI Circuits}, Kyoto, JP, Dec.
  2011, pp. 162--163.

\bibitem{QianGaoAmBCTWC16}
J.~Qian, F.~Gao, G.~Wang, S.~Jin, and H.~Zhu, ``Noncoherent detections for
  ambient backscatter system,'' \emph{{IEEE} Trans. Wireless Commun.}, vol.~16,
  no.~3, pp. 1412--1422, Mar. 2017.

\bibitem{WangGaoAmBCTCOM16}
G.~Wang, F.~Gao, R.~Fan, and C.~Tellambura, ``Ambient backscatter communication
  systems detection and performance analysis,'' \emph{{IEEE} Trans. Commun.},
  vol.~64, no.~11, pp. 4836--4846, Nov. 2016.

\bibitem{WiFiBackscatter14}
B.~Kellogg, A.~Parks, S.~Gollakota, J.~R. Smith, and D.~Wetherall, ``{Wi-Fi}
  backscatter: Internet connectivity for {RF}-powered devices,'' in \emph{Proc.
  of ACM SIGCOMM}, Chicago, USA, Jun. 2014, pp. 1--12.

\bibitem{TurbochargingABCSigcom14}
A.~N. Parks, A.~Liu, S.~Gollakota, and J.~R. Smith, ``Turbocharging ambient
  backscatter communication,'' in \emph{Proc. of ACM SIGCOMM}, Chicago, IL,
  USA, Aug. 2014, pp. 1--12.

\bibitem{BackFiSigcom15}
D.~Bharadia, K.~Joshi, M.~Kotaru, and S.~Katti, ``{BackFi}: High throughput
  {WiFi} backscatter,'' in \emph{Proc. of ACM SIGCOMM}, London, UK, Aug. 2015,
  pp. 283--296.

\bibitem{HitchHikeKattiSenSys16}
P.~Zhang, D.~Bharadia, K.~Joshi, and S.~Katti, ``Hitchhike: Practical
  backscatter using commodity {WiFi},'' in \emph{Proc. of ACM Conf.Embedded
  Netw. Sensor Sys.}, Stanford, CA, USA, Nov. 2016, pp. 259--271.

\bibitem{PassiveWiFiNSDI16}
B.~Kellogg, V.~Talla, S.~Gollakota, and J.~R. Smith, ``Passive {Wi-Fi}:
  Bringing low power to {Wi-Fi} transmissions,'' in \emph{Proc. of {USENIX}
  Symposium on Networked Systems Design and Implementation (NSDI)}, Santa
  Clara, CA, USA, Mar. 2016, pp. 151--164.

\bibitem{InterscatterSigcom16}
V.~Iyery, V.~Tallay, B.~Kelloggy, S.~Gollakota, and J.~R. Smith,
  ``Inter-technology backscatter: Towards internet connectivity for implanted
  devices,'' in \emph{Proc. of ACM SIGCOMM}, Florianopolis, Brazil, Aug. 2016,
  pp. 356--369.

\bibitem{FSBackscatterSigcomm16}
P.~Zhang, M.~Rostami, P.~Hu, and D.~Ganesan, ``Enabling practical backscatter
  communication for on-body sensors,'' in \emph{Proc. of ACM SIGCOMM},
  Florianopolis, Brazil, Aug. 2016, pp. 370--383.

\bibitem{HeWangCL11}
C.~He and Z.~J. Wang, ``Closed-form {BER} analysis of non-coherent {FSK} in
  {MISO} double rayleigh fading/{RFID} channel,'' \emph{{IEEE} Commun. Lett.},
  vol.~15, no.~8, pp. 848--850, Aug. 2011.

\bibitem{BoyerSumit14}
C.~Boyer and S.~Roy, ``Backscatter communication and {RFID}: coding, energy,
  and {MIMO} analysis,'' \emph{{IEEE} Trans. Commun.}, vol.~62, no.~3, pp.
  770--785, Mar. 2014.

\bibitem{Carvalho14}
N.~B. Carvalho and etc., ``Wireless power transmission: R$\&${D} activities
  within europe,'' \emph{{IEEE} Trans. Microwave Theory Tech.}, vol.~62, no.~4,
  pp. 1031--1045, Apr. 2014.

\bibitem{BoaventuraCarvalho13}
A.~J.~S. Boaventura and N.~B. Carvalho, ``Extending reading range of commercial
  {RFID} readers,'' \emph{{IEEE} Trans. Microwave Theory Tech.}, vol.~61,
  no.~1, pp. 633--640, Jan. 2013.

\bibitem{YangBackscatter15}
G.~Yang, C.~K. Ho, and Y.~L. Guan, ``Multi-antenna wireless energy transfer for
  backscatter communication systems,'' \emph{{IEEE} J. Sel. Areas Commun.},
  vol.~33, no.~12, pp. 2974--2987, Dec. 2015.

\bibitem{KimionisSahalosTCOM14}
J.~Kimionis, A.~Bletsas, and J.~N. Sahalos, ``Increased range bistatic scatter
  radio,'' \emph{{IEEE} Trans. Commun.}, vol.~62, no.~3, pp. 1091--1104, Mar.
  2014.

\bibitem{HuangWPBCN16}
K.~Han and K.~Huang, ``Wirelessly powered backscatter communication network:
  Modeling, coverage and capacity,'' to appear in {\it {IEEE} Trans. Wireless
  Commun.} [Early Access].

\bibitem{GoldsmithWC2005}
A.~Goldsmith, \emph{Wireless Communications}.\hskip 1em plus 0.5em minus
  0.4em\relax Cambridge Univ. Press, 2005.

\bibitem{BryantRFIC2012}
C.~Bryant and H.~Sj\"{o}land, ``A 2.45ghz ultra-low power quadrature front-end
  in 65nm {CMOS},'' in \emph{Proc. of {IEEE} Radio Frequency Integrated
  Circuits Symposium}, Montreal, QC, Canada, Jun. 2012, pp. 247--250.

\bibitem{BoyerRoyTWC13}
C.~Boyer and S.~Roy, ``Space time coding for backscatter {RFID},'' \emph{{IEEE}
  Trans. Wireless Commun.}, vol.~12, no.~5, pp. 2272--2280, May 2013.

\bibitem{ZhangLiangJSTSP15}
R.~Zhang and Y.-C. Liang, ``Exploiting multi-antennas for opportunistic
  spectrum sharing in cognitive radio networks,'' \emph{{IEEE} J. Select.
  Topics Signal Process.}, vol.~2, no.~1, pp. 88--102, Sept. 2008.

\bibitem{YCLiangTradeoffTWC08}
Y.-C. Liang, Y.~Zeng, E.~C.~Y. Peh, and A.~T. Hoang, ``Sensing-throughput
  tradeoff for cognitive radio networks,'' \emph{{IEEE} Trans. Wireless
  Commun.}, vol.~7, no.~4, pp. 1326--1336, Apr. 2008.

\bibitem{SMKayStatisticalSP93}
S.~M. Kay, \emph{Fundamentals of Statistical Signal Processing}.\hskip 1em plus
  0.5em minus 0.4em\relax NJ, USA: Prentice Hall, 1993.

\bibitem{WalpoleProbbook2011}
R.~E. Walpole, R.~H. Myers, S.~L. Myers, and E.~Ye, \emph{Probability and
  Statistics for Engineers and Scientists}.\hskip 1em plus 0.5em minus
  0.4em\relax Pearson, 9th edition.

\bibitem{JGProakisMSalehi05}
J.~G. Proakis and M.~Salehi, \emph{Digital Communications (5th Edition)}.\hskip
  1em plus 0.5em minus 0.4em\relax NY, USA: McGraw-Hill, 2007.

\bibitem{80211a}
IEEE Std 802.11a - 1999, Wireless LAN Medium Acccess Control (MAC) and Physical
  Layer (PHY) Specifications[S], 1999.

\bibitem{MorelliKuoProceed07}
M.~Morelli, C.-C.~J. Kuo, and M.-O. Pun, ``Synchronization techniques for
  orthogonal frequency division multiple access ({OFDMA}): A tutorial review,''
  \emph{Proceedings of the IEEE}, vol.~95, no.~7, pp. 1394--1427, Jul. 2007.

\end{thebibliography}
\bibliographystyle{IEEEtran}%By using IEEEtrans, the number can be displayed.

\end{document}